%% file: TR.tex
\documentclass{sig-alternate-2013}

\usepackage{url}
\usepackage{xspace}
\usepackage{color}
\usepackage{graphicx}
\usepackage{mathpartir}

\makeatletter
\renewcommand*{\@opargbegintheorem}[3]{\trivlist
      \item[\hskip \labelsep{\bfseries #1\ #2}] \textbf{(#3)}\ \itshape}
\makeatother

\newcommand{\framework}{\ensuremath{\mathsf{Eunomia}}\xspace}

\newcommand{\dt}{\ensuremath{\framework^\mathsf{DET}}\xspace}
\newcommand{\kh}{\ensuremath{\framework^\mathsf{KH}}\xspace}

\begin{document}

\title{Equivalence-based Security for Querying Encrypted Databases: Theory and Application to Privacy Policy Audits}

\numberofauthors{3} 
%
\author{
%
%
\alignauthor
Omar Chowdhury\\
       \affaddr{Purdue University}\\
       \affaddr{West Lafayette, Indiana}\\
       \email{ochowdhu@purdue.edu}
\alignauthor
Deepak Garg\\
       \affaddr{MPI-SWS}\\
       \affaddr{Saarbr\"{u}cken, Germany}\\
       \email{dg@mpi-sws.org}
\alignauthor
Limin Jia, Anupam Datta\\
       \affaddr{Carnegie Mellon University}\\
       \affaddr{Pittsburgh, Pennsylvania}\\
       \email{\{liminjia,danupam\}@cmu.edu}
}

\maketitle              

\input{command}

\begin{abstract}
Motivated by the problem of simultaneously preserving confidentiality
and usability of data outsourced to third-party clouds, we present two
different database encryption schemes that largely hide data but
reveal enough information to support a wide-range of relational
queries. We provide a security definition for database encryption that
captures confidentiality based on a notion of equivalence of databases
from the adversary's perspective. As a specific application, we adapt
an existing algorithm for finding violations of privacy policies to
run on logs encrypted under our schemes and observe low to moderate
overheads.
\end{abstract}




\section{Introduction}
\label{section:introduction}
\input{intro2}

\section{Overview of Eunomia}
\label{section:motivation}
\input{motivation}

 \section{Policy and Log Specifications}
 \label{section:spec}
 \input{spec}

 \section{Encryption Schemes}
 \label{section:functions}
 \input{function}

\section{Security Analysis}
\label{section:securityAnalysis}
\input{security2}

\section{Auditing Algorithm}
\label{section:algorithm}
\input{algorithm}

\section{EQ Mode Check}
\label{section:modeChecking}
\input{mode}



\section{Implementation and Evaluation}
\label{section:implementation}
\input{implementation2}

\section{Related Work}
\label{section:relatedWork}
\input{related1}

\section{Summary}
\label{section:conclusion}
\input{conclusion}

\bibliographystyle{abbrv}


\end{document}

%% file: command.tex
\newcommand{\notes}[1]{\textcolor{red}{#1}}
\newcommand{\dg}[1]{\notes{Deepak says: #1}}
\newcommand{\limin}[1]{\notes{Limin says: #1}}
\newcommand{\anupam}[1]{\notes{Anupam says: #1}}
\newcommand{\omar}[1]{\notes{Omar says: #1}}

\newcommand{\eat}[1]{}
\newcommand{\Paragraph}[1]{\vspace{2pt}\noindent\textbf{\textit{#1}}}
\newcommand{\evalPoint}{\ensuremath{\mathit{evalPtr}}\xspace}
\newcommand{\curPoint}{\ensuremath{\mathit{curPtr}}\xspace}

\newcommand{\etal}{\textit{et al.}\xspace}
\newcommand{\ie}{\textit{i.e.}\xspace}
\newcommand{\eg}{\textit{e.g.}\xspace}
\newcommand{\viz}{\textit{viz.}\xspace}
\newcommand{\etc}{\textit{etc.}\xspace}

\newcommand{\fe}{\pred{FE}\xspace}
\newcommand{\var}{\pred{Var}}
\newcommand{\const}{\pred{Const}\xspace}
\newcommand{\arity}[1]{\ensuremath{\alpha(\pred{#1})}\xspace}

\newcommand{\cT}{\ensuremath{\mathcal{T}}\xspace}

\newcommand{\state}{\ensuremath{\pi}\xspace}
\newcommand{\statedown}[1]{\ensuremath{\pi_{#1}}\xspace}
\newcommand{\stateup}[1]{\ensuremath{\pi^{#1}}\xspace}
\newcommand{\structure}[1]{\ensuremath{\mathbb{S}_{#1}}\xspace}
\newcommand{\structureud}[2]{\ensuremath{\mathbb{S}_{#1}^{#2}}\xspace}
\newcommand{\idx}{\ensuremath{idx}\xspace}
\newcommand{\map}{\ensuremath{\mathcal{A}}\xspace}
\newcommand{\result}{\ensuremath{\mathbb{R}}\xspace}

\newcommand{\yesterday}{\ensuremath{\LTLcircleminus}\xspace}
\newcommand{\tomorrow}{\ensuremath{\LTLcircle}\xspace}
\newcommand{\henceforth}{\ensuremath{\LTLsquare}\xspace}
\newcommand{\eventually}{\ensuremath{\LTLdiamond}\xspace}
\newcommand{\historically}{\ensuremath{\LTLsquareminus}\xspace}
\newcommand{\once}{\ensuremath{\LTLdiamondminus}\xspace}
\newcommand{\since}{\ensuremath{\,\mathrm{\cal S}\,}\xspace}
\newcommand{\until}{\ensuremath{\,\mathrm{\cal U}\,}\xspace}
\newcommand{\weaksince}{\ensuremath{\,\mathrm{\cal B}\,}\xspace}
\newcommand{\weakuntil}{\ensuremath{\,\mathrm{\cal W}\,}\xspace}
\newcommand{\weakyesterday}{\ensuremath{\LTLcircletilde}}

\newcommand{\TIME}{\ensuremath{\tau}\xspace}
\newcommand{\LOG}{\ensuremath{\mathcal{L}}\xspace}
\newcommand{\PF}{\ensuremath{\psi}\xspace}
\newcommand{\FS}{FSOV}
\newcommand{\BS}{BuildTemporalStructure}
\newcommand{\FSWV}{FSWV}

\newcommand{\cur}{\ensuremath{\chi_C}\xspace}
\newcommand{\fut}{\ensuremath{\chi_F}\xspace}
\newcommand{\bld}{\textbf{B}\xspace}
\newcommand{\nbld}{\textbf{NB}\xspace}
\newcommand{\chk}{\textbf{C}\xspace}
\newcommand{\tl}{\textbf{TL}\xspace}
\newcommand{\ut}{\textbf{IT}\xspace}
\newcommand{\tc}{\textbf{tc}\xspace}
\newcommand{\tb}{\textbf{tb}\xspace}
\newcommand{\na}{\textbf{na}\xspace}
\newcommand{\ms}[1]{\ensuremath{\{#1\}}\xspace}

\newcommand{\mymap}{\ensuremath{\delta}\xspace}
\newcommand{\bsat}{\ensuremath{\vdash_{\textbf{\color{blue}g}}}\xspace}
\newcommand{\nsat}{\ensuremath{\vdash}\xspace}

\newcommand{\bformula}{\texttt{B-formula}\xspace}
\newcommand{\bformulas}{\texttt{B-formulas}\xspace}

\newcommand{\funcname}[1]{\textbf{\texttt{#1}}}

\newcommand{\cc}{\funcname{checkCompliance}\xspace} 
\newcommand{\bts}{\funcname{uSS}\xspace} 
\newcommand{\ips}{\funcname{ips}\xspace} 
\newcommand{\sat}{\funcname{esat}\xspace} 
\newcommand{\dom}{\funcname{dom}\xspace} 
\newcommand{\monitor}{\funcname{pr\'ecis}\xspace}

\newcommand{\statesize}{\ensuremath{\Upsilon}\xspace}

\newcommand{\tsub}{\ensuremath{\mathsf{b\mbox{-}tsub}}\xspace}
\newcommand{\stsub}{\ensuremath{\mathsf{b\mbox{-}s\mbox{-}tsub}}\xspace}

\newcommand{\ins}{\ensuremath{\sigma_{\mathrm{in}}}\xspace}
\newcommand{\DELAY}{\ensuremath{\Delta}\xspace}
\newcommand{\delay}{\ensuremath{\delta}\xspace}
\newcommand{\invalid}{\ensuremath{\sigma_\bot}\xspace}

\newcommand{\lsat}{\ensuremath{\vdash_{\textbf{\color{magenta}tl}}}\xspace}

\newcommand{\inp}{\ensuremath{\chi_I}\xspace}
\newcommand{\outp}{\ensuremath{\chi_O}\xspace}

\newcommand{\SFEC}{\funcname{EA}\xspace}

\newcommand{\ereduce}{\pred{\mathbf{ereduce}}}
\newtheorem{theorem}{Theorem}
\newtheorem{lemma}{Lemma}
\newtheorem{axiom}{Axiom}
\newtheorem{corollary}{Corollary}
\newtheorem{definition}{Definition}
\newtheorem{example}{Example}
\newtheorem{myclaim}{Claim}

\newcommand{\emakeNode}{\pred{MakeNode}}
\newcommand{\makeNode}[1]{\ensuremath{\mathsf{MakeNode}(\langle#1\rangle)}}
\newcommand{\addChild}[1]{\ensuremath{\pred{AddChild}(#1)}}
\newcommand{\keys}{\ensuremath{\mathcal{K}}\xspace}
\newcommand{\escheme}{\ensuremath{\delta}\xspace}
\newcommand{\cM}{\ensuremath{\mathcal{M}}\xspace}
\newcommand{\cB}{\ensuremath{\mathcal{B}}\xspace}
\newcommand{\ets}{\pred{encryptedTimeStamps}}
\newcommand{\cS}{\ensuremath{\mathcal{S}}\xspace}

\newcommand{\CPLA}{\ensuremath{\mathsf{IND\mbox{-}CPLA^{DET}}}\xspace}
\newcommand{\CPLAKH}{\ensuremath{\mathsf{IND\mbox{-}CPLA^{KH}}}\xspace}

\newcommand{\SQECPLA}{\ensuremath{\mathsf{IND\mbox{-}SQECPLA}}\xspace}
\newcommand{\ECPLA}{\ensuremath{\mathsf{IND\mbox{-}ECPLA}}\xspace}
\newcommand{\OCPA}{\ensuremath{\mathsf{IND\mbox{-}OCPA}}\xspace}
\newcommand{\DOCPA}{\ensuremath{\mathsf{IND\mbox{-}CDDA}}\xspace}
\newcommand{\LOR}{\ensuremath{\mathsf{LoR}}\xspace}
\newcommand{\DOPE}{\ensuremath{\mathsf{\mathbf{mOPED}}}\xspace}
\newcommand{\MOPE}{\ensuremath{\mathsf{\mathbf{mOPE}}}\xspace}
\newcommand{\cl}{\ensuremath{\mathsf{Cl}}\xspace}
\newcommand{\adv}{\ensuremath{\mathsf{Adv}}\xspace}
\newcommand{\sk}{\ensuremath{\mathsf{SK}}\xspace}
\newcommand{\param}{\ensuremath{\kappa}\xspace}
\newcommand{\server}{\ensuremath{\pred{Ser}}\xspace}
\newcommand{\encrypt}{\ensuremath{\pred{ENCRYPT}}\xspace}
\newcommand{\decrypt}{\ensuremath{\pred{DECRYPT}}\xspace}
\newcommand{\keygen}{\ensuremath{\mathsf{KeyGen}}\xspace}
\newcommand{\deter}{\ensuremath{\mathsf{DET}}\xspace}

\newcommand{\CS}{\ensuremath{\mathsf{CS}}\xspace}

\newcommand{\sigin}{\ensuremath{\sigma_{\textit{in}}}\xspace}
\newcommand{\Sigout}{\ensuremath{\Sigma_{\textit{out}}}\xspace}
\newcommand{\JJoin}{\hbox{{\Large$\bowtie$}}\xspace}

\newcommand{\planguage}{\ensuremath{\mathbf{\mathcal{P}_{\mathcal{L}}}}\xspace}

\newcommand{\interval}{\ensuremath{\mathbb{I}}\xspace}

\newcommand{\cL}{\ensuremath{\mathcal{L}}\xspace}
\newcommand{\cE}{\ensuremath{\mathcal{E}}\xspace}
\newcommand{\cD}{\ensuremath{\mathcal{D}}\xspace}
\newcommand{\cF}{\ensuremath{\mathcal{F}}\xspace}

\newcommand{\policy}{\ensuremath{\varphi}\xspace}

\newcommand{\identity}{\ensuremath{\mathbf{\bullet}}\xspace}




\newcommand{\send}{\ensuremath{\mathrm{send}}\xspace}
\newcommand{\inrole}{\ensuremath{\mathrm{inrole}}\xspace}
\newcommand{\contains}{\ensuremath{\mathrm{contains}}\xspace}
\newcommand{\psychotherapynotes}{\ensuremath{\mathit{psych\mbox{-}notes}}\xspace}
\newcommand{\PHI}{\ensuremath{\mathit{PHI}}\xspace}
\newcommand{\coveredentity}{\ensuremath{\mathit{covered\mbox{-}entity}}\xspace}
\newcommand{\individual}{\ensuremath{\mathit{individual}}\xspace}


\newcommand{\bbbt}{\ensuremath{\mathbb{T}}\xspace}
\newcommand{\bbbr}{\ensuremath{\mathbb{R}}\xspace}
\newtheorem{proposition}{Proposition}
\newtheorem{claim}{Claim} 
%
%

%
%

\newcommand{\pred}[1]{\ensuremath{\mathsf{#1}}\xspace}

\newcommand{\db}{\ensuremath{\mathsf{DB}}\xspace}
\newcommand{\dbd}[1]{\ensuremath{\mathsf{DB}_{#1}}\xspace}
\newcommand{\dbu}[1]{\ensuremath{\mathsf{DB}^{#1}}\xspace}
\newcommand{\dbud}[2]{\ensuremath{\mathsf{DB}_{#1}^{#2}}\xspace}

\newcommand{\tbl}{\ensuremath{\mathsf{table}}\xspace}
\newcommand{\tbld}[1]{\ensuremath{\mathsf{table}_{#1}}\xspace}
\newcommand{\tblu}[1]{\ensuremath{\mathsf{table}^{#1}}\xspace}
\newcommand{\tblud}[2]{\ensuremath{\mathsf{table}_{#1}^{#2}}\xspace}

\newcommand{\sch}{\ensuremath{\mathsf{Tschema}}\xspace}
\newcommand{\schd}[1]{\ensuremath{\mathsf{Tschema}_{#1}}\xspace}
\newcommand{\schu}[1]{\ensuremath{\mathsf{Tschema}^{#1}}\xspace}
\newcommand{\schud}[2]{\ensuremath{\mathsf{Tschema}_{#1}^{#2}}\xspace}

\newcommand{\rw}{\ensuremath{\mathsf{row}}\xspace}
\newcommand{\rwd}[1]{\ensuremath{\mathsf{row}_{#1}}\xspace}
\newcommand{\rwu}[1]{\ensuremath{\mathsf{row}^{#1}}\xspace}
\newcommand{\rwud}[2]{\ensuremath{\mathsf{row}_{#1}^{#2}}\xspace}

\newcommand{\elem}{\ensuremath{\mathsf{element}}\xspace}
\newcommand{\elemd}[1]{\ensuremath{\mathsf{element}_{#1}}\xspace}
\newcommand{\elemu}[1]{\ensuremath{\mathsf{element}^{#1}}\xspace}
\newcommand{\elemud}[2]{\ensuremath{\mathsf{element}_{#1}^{#2}}\xspace}

\newcommand{\gtoc}[1]{\ensuremath{\mathsf{get\_time\_offset\_consts}(#1)}}
\newcommand{\gc}[1]{\ensuremath{\mathsf{get\_consts}(#1)}}
\newcommand{\gv}[1]{\ensuremath{\mathsf{get\_vars}(#1)}}
\newcommand{\struct}[1]{\ensuremath{\mathcal{ST}(#1)}}

\newcommand{\sathat}{\ensuremath{\widehat{\pred{esat}}}\xspace}
\newcommand{\reduce}{\ensuremath{\pred{\mathbf{reduce}}}\xspace}
\newcommand{\st}[3]{\ensuremath{\langle#1, #2, #3\rangle}\xspace}
\newcommand{\sts}[4]{\ensuremath{\langle#1, #2, #3, #4\rangle}\xspace}
\newcommand{\act}[2]{\ensuremath{\overset{\langle#1, #2\rangle}{\Longrightarrow}}\xspace}
\newcommand{\acts}[2]{\ensuremath{\overset{\langle#1, #2\rangle}{\Longrightarrow\hspace{-0.05in}*}}\xspace}
\newcommand{\Act}[2]{\ensuremath{\overset{\langle#1, #2\rangle}{ \rightsquigarrow}}\xspace}
\newcommand{\Acts}[2]{\ensuremath{\overset{\langle#1, #2\rangle}{ \rightsquigarrow\hspace{-0.05in}*}}\xspace}
\newcommand{\ir}{\ensuremath{\mathbb{I}}\xspace}

\newenvironment{packedenumerate}{\vspace{-5pt}
\begin{enumerate}
\setlength{\topsep}{0pt}
\setlength{\itemsep}{0pt}
\setlength{\partopsep}{0pt}
}{\end{enumerate}}

\newenvironment{packeditemize}{\vspace{-5pt}
\begin{itemize}
\setlength{\topsep}{0pt}
\setlength{\itemsep}{0pt}
\setlength{\partopsep}{0pt}
}{\end{itemize}}

%% file: intro2.tex
To reduce infrastructure costs, small- and medium-sized businesses may
outsource their databases and database applications to third-party
clouds. However, proprietary data is often private, so storing it in a
cloud raises confidentiality concerns. Client-side encryption of
databases prior to oursourcing alleviates confidentiality concerns,
but it also makes it impossible to run any relational queries on the
outsourced databases. Several prior research projects have
investigated encryption schemes that trade-off perfect data
confidentiality for the ability to run relational
queries~\cite{SWP00,BS11,HILM02}. However, these schemes either
require client-side processing~\cite{HILM02}, or require additional
hardware support~\cite{BS11}, or support a very restrictive set of
queries~\cite{SWP00}. Our long-term goal is to develop database
encryption schemes that can (1) run on commodity off-the-shelf (COTS)
cloud infrastructure without any special hardware or any kernel
modifications, (2) support a broad range of relational queries on the
encrypted database without interaction with the client, and (3)
provide provable end-to-end security and a precise characterization of
what information encryption leaks for a given set of queries. Both in
objective and in method, our goal is similar to that of
CryptDB~\cite{cryptDB}, which attains properties (1) and (2), but not
(3).

As a step towards our goal, in this paper, we design two database
encryption schemes, \dt and \kh, with properties (1), (2) and (3). Our
design is guided by, and partly specific to, a single application
domain, namely, audit of data-use logs for violations of privacy
policies. This application represents a real-world problem.
For example, in the US, the healthcare and finance industry must
handle client data in accordance with the privacy portions of the
federal acts Health Insurance Portability and Accountability Act
(HIPAA) \cite{HIPAA} and Gramm-Leach-Bliley Act (GLBA) \cite{GLBA}
respectively, in addition to state legislation. To remain compliant
with privacy legislation, organizations record logs of
privacy-relevant day-to-day operations such as data access/use and
employee role changes, and audit these logs for violations of privacy
policies, either routinely or on a case-by-case basis.  Logs can get
fairly large and are often organized in commodity databases. Further,
audit is computationally expensive but highly parallelizable, so there
is significant merit in outsourcing the storage of logs and the
execution of audit algorithms to third-party clouds.

\vspace{5mm}
\noindent\textbf{Generality.}
Audit is a challenging application for encryption schemes because
audit requires almost all relational query operations on logs. These
operations include selection, projection, join, comparison of fields,
and what we call
\emph{displaced comparison} (is the difference between two fields less than a
given constant?). Both our encryption schemes support all these query
operations. The only standard query operation not commonly required by
privacy audit (and not supported by our schemes) is aggregation (sums
and averages; counting queries are supported). Any application that
requires only the query operations listed above can be adapted to run
on \dt or \kh, even though this paper focuses on the audit application
only.

{\dt} and \kh trade efficiency and flexibility in supported queries
differently. \dt uses deterministic encryption and has very low
overhead (3\% to 9\% over a no encryption baseline in our audit
application), but requires anticipating prior to encryption which
pairs of columns will be join-ed in audit queries. \kh uses
Popa \etal's adjustable key hash scheme~\cite{ADJJOIN,cryptDB} for
equality tests and has higher overhead (63\% to 406\% in our audit
application), but removes the need to anticipate join-ed columns
ahead-of-time.

\vspace{5mm}
\noindent\textbf{Security.}
We characterize formally what information about an encrypted log
(database) our schemes may leak to a PPT adversary with direct access
to the encrypted store (modeling a completely adversarial cloud). We
prove that by looking at a log encrypted with either of our schemes,
an adversary can learn (with non-negligible probability) only that the
plaintext log lies within a certain, precisely
defined \emph{equivalence class of logs}.  This class of logs
characterizes the uncertainty of the adversary and, hence, the
confidentiality of the encrypted log~\cite{askarov07}. Prior work like
CryptDB lacks such a theorem. CryptDB uses a trusted proxy server to
dynamically choose the most secure encryption scheme for every
database column (from a pre-determined set of schemes), based on the
queries being run on that column. While each scheme is known to be
secure in isolation and it is shown that at any time, a column is
encrypted with the weakest scheme that supports all past queries on
the column~\cite[Theorem 2]{...}, there is no end-to-end
characterization of information leaked after a sequence of
queries. (In return, CryptDB supports all SQL-queries, including
aggregation queries, which we do not.)

\vspace{5mm}
\noindent\textbf{Functionality.}
To demonstrate that our encryption schemes support nontrivial
applications, we adapt an audit algorithm called \reduce from our
prior~\cite{GJD11} to run on logs encrypted with either \dt or \kh. We
implement and test the adapted algorithm, \ereduce, on both schemes
and show formally that it is functionally correct
on \dt. Since \ereduce exercises all the log-query operations listed
above, this is strong evidence that our schemes support those
operations correctly. To run \ereduce on \dt (\kh) we need to know
prior to encryption (audit) which columns in the log will be compared
for equality or joined. To this end, we develop a new static analysis
of policies, which we call the EQ mode check.

Privacy audit often requires comparisons of the form ``is timestamp t1
within 30 days of timestamp t2''? We call such comparisons ``displaced
comparisons'' (30 days is called the ``displacement'').  To support
displaced comparisons, we design and prove the security of a new
cryptographic sub-scheme dubbed \DOPE (mutable order-preserving
encoding with displacement). This scheme extends the \MOPE scheme of
Popa \etal~\cite{mope}, which does not support displacements, and may
be of independent interest.

\vspace{5mm}
\noindent\textbf{Deployability.}
Both \dt and \kh can be deployed on commodity cloud database systems
with some additional metadata. In both schemes, the client encrypts
individual data cells locally and stores the ciphertexts in a
commodity database system in the cloud (possibly incrementally). Audit
runs on the cloud without interaction with the client and returns
encrypted results to the client, which can decrypt them. The schemes
reveal enough information about the data to perform all supported
operations, e.g., compare two data values for equality to perform
equi-joins.

\vspace{5mm}
\noindent\textbf{Contributions.}
We make the following technical contributions:
\begin{packeditemize}\setlength{\itemsep}{-2pt}
\item
We introduce two database encryption schemes, \dt and \kh, that
support selection, projection, join, comparison of fields, and
displaced comparison queries. The schemes trade efficiency for the
need to predict expected pairs of join-ed columns before
encryption. As a building block, we develop the sub-scheme \DOPE, that
allows displaced comparison of encrypted values.
\item
We characterize confidentiality preserved by our schemes as
equivalence classes of plaintext logs and prove that both our schemes
are end-to-end secure.
\item
We adapt an existing privacy policy audit algorithm to execute on our
schemes. We prove the functional correctness of the execution of our
algorithm on \dt.
\item
We implement both our schemes and the adapted audit algorithm,
observing low overheads on \dt and moderate overheads on \kh.
\end{packeditemize}


\vspace{5mm}
\noindent\textbf{Notation.}
This paper is written in the context of the privacy audit application
and our encryption schemes are presented within this context.
Accordingly, we sometimes use the term ``log'' or ``audit log'' when
the more general term ``database'' could have been used and,
similarly, use the term ``policy'' or ``privacy policy'' when the more
general term ``query'' could have been used. We expect our schemes to
generalize to other database applications straightforwardly.

%% file: motivation.tex
We first present the architecture of \framework.  Then, we motivate
our choice of encryption schemes through examples and discuss
policy audit in \framework in more detail.  Finally, we discuss our
goals, assumptions, and adversary model.

\subsection{Architecture of Eunomia}
We consider the scenario where an organization, called the client
or \cl, with sensitive data and audit requirements wishes to outsource
its log (organized as a relational database) and audit process
(implemented as a sequence of policy-dependent queries) to a
potentially compromisable third-party cloud server, denoted \CS.
\cl generates the log
from its day-to-day operations. \cl then encrypts the log and
transfers the encrypted log to the \CS.
\cl initiates the audit process by choosing a policy. The auditing algorithm
 runs on the \CS infrastructure and the audit result containing
 encrypted values is sent to \cl, which can decrypt the values in the
 result.

The mechanism of log generation is irrelevant for us. From our
perspective, a log is a pre-generated database with a public schema,
where each table defines a certain privacy-relevant predicate. For
example, the table \pred{Roles} may contain columns \pred{Name} and
\pred{Role}, and may define the mapping of \cl's employees to
\cl's organizational roles. Similarly, the table \pred{Sensitive\_accesses}
may contain columns \pred{Name}, \pred{File\_name}, and \pred{Time},
recording who accessed which sensitive file at what time.  Several
tables like \pred{Sensitive\_accesses} may contain columns with
timestamps, which are integers.

\subsection{Encryption Schemes}
%
An organization may na\"{i}vely encrypt the entire log with a strong
encryption scheme before transferring it to a cloud, but this renders
the stored log ineffective for audit, as audit (like most other
database computations) must \emph{relate} different parts of the
log. For example, suppose the log contains two tables $T_1$ and $T_2$.
$T_1$ lists the names of individuals who accessed patients'
prescriptions. $T_2$ lists the roles of all individuals in the
organization. Consider the privacy policy:

\vspace*{5pt}\noindent\textbf{Policy 1:} \textit{Every individual accessing patients' prescriptions must
be in the role of Doctor}.\vspace*{2pt}

%
The audit process of the above policy must read names from $T_1$ and
test them for equality against the list of names under the role Doctor
in $T_2$. This forces the use of an encryption method that allows
equality tests (or equi-joins). Unsurprisingly, this compromises the
confidentiality of the log, as an adversary (e.g., the cloud host,
which observes the audit process) can detect equality between
encrypted fields (e.g., equality of names in $T_1$ and $T_2$).
However, not all is lost: for instance, if per-cell deterministic
encryption is used, the adversary cannot learn the concrete names
themselves.

A second form of data correlation necessary for audit is the order
between time points. Consider the following policy:

\vspace*{5pt}\noindent\textbf{Policy 2:} \textit{If an outpatient's medical record is accessed by an
employee of the Billing Department, then the outpatient must have
visited the medical facility in the last one month}.\vspace*{2pt}

Auditing this policy requires checking whether the distance between
the timestamps in an access table and the timestamps in a patient
visit table is shorter than a month. In this case, the encryption
scheme must reveal not just the relative order of two timestamps but
also the order between a timestamp and another timestamp displaced by
one month.  Similar to Policy 1, the encryption scheme must reveal
equality between patient names in the two
tables. 

To strike a balance between functional (audit) and confidentiality
requirements, we
investigate two cryptographic schemes, namely \dt and \kh, to encrypt
logs.
Each cell in the database tables is encrypted
individually.  All cells in a column are encrypted using the same key.
\dt uses deterministic encryption to support equality tests; two
columns that might be tested for equality by subsequent queries are
encrypted with the same key.  \dt requires that log columns that might
be tested for equality during audit are known prior to the encryption.
Audit under \dt is quite efficient. However, adapting encrypted logs
to audit different policies that require different column equality
tests requires log re-encryption, which is costly.
Our second scheme \kh handles frequent policy updates efficiently.
\kh relies on the adjustable key hash
scheme~\cite{ADJJOIN,cryptDB} for equality tests. A transfer token is
generated for each pair of columns needed to be tested for equality
prior to audit.
\kh additionally stores keyed hashes of all cells.
 Audit under \kh requires the audit algorithm to track the provenance
of the ciphertext (i.e., from which table, which column the ciphertext
originated) and is less efficient than audit under \dt.
%

To support timestamp comparison with displacement (shown in Policy 2)
the \DOPE scheme, described in Section~\ref{section:functions}, is
used by both \dt and \kh.  Like its predecessor, \MOPE~\cite{mope},
the scheme adds an additional search tree (additional metadata) to the
encrypted database on \CS. (Supporting displacements is necessary for
a practical audit system because privacy regulations use displacements
to express obligation deadlines.  Out of 84 HIPAA privacy clauses, 7
use displacements. Cignet Health of Prince George's County, Maryland
was fined \$1.3 million for violating one of these
clauses, \S164.524.)

The encrypted database has a schema derived from the schema of the
plaintext database and may be stored on \CS using any standard
database management systems (DBMS). The DBMS may be used to index the
encrypted cells.  As shown in~\cite{GJD11}, database indexing plays a
key role in improving the efficiency of the audit process.  Hence, we
develop our encryption scheme in such a way that it is possible to
leverage database indexing supported by commodity DBMS.




\subsection{Policies and audit}

Privacy policies may be
extracted from privacy legislation like HIPAA \cite{HIPAA} and GLBA \cite{GLBA},
or from internal company
requirements. Technically, a privacy policy specifies a
\emph{constraint} on the log. For example, Policy~1 of
Section~\ref{section:introduction}
requires that any name appearing in table $T_1$ appear in table $T_2$
with role Doctor. Generally, policies can be complex and may mention
various entities, roles, time, and subjective beliefs. For instance,
DeYoung \etal's formalization of the HIPAA and GLBA Privacy Rules span
over 80 and 10 pages, respectively~\cite{DGJKD10}.
%
%
We represent policies as formulas of first-order logic (FOL) because we find
it technically convenient and because FOL has been
demonstrated in prior work to be adequate for representing policies
derived from existing privacy legislation (DeYoung \etal, mentioned
above, use the same representation). We describe this logic-based
representation of policies in Section~\ref{section:spec}.

Our audit algorithm adapts our prior algorithm, \reduce~\cite{GJD11},
that works on policies represented in FOL.  This algorithm takes as
input a policy and a log and \emph{reduces} the policy by checking the
policy constraints on the log. It outputs constraints that cannot be
checked due to lack of information (missing log tables, references to
future time points, or need for human intervention) in the form of
a \emph{residual policy}. Similar to \reduce, our adapted
algorithm, \ereduce, uses database queries as the basic building
block.  Our encryption schemes permit queries with selection,
projection, join, comparison and displaced comparison operations.  Our
schemes do not support queries like aggregation (which would require
an underlying homomorphic encryption scheme and completely new
security proofs).

To run \reduce on \dt, we need to identify columns that are tested for
equality.  This information is needed prior to encryption for \dt and
prior to audit for \kh, as explained in
Section~\ref{section:functions}. We develop a static analysis of
policies represented in FOL, which we call the \emph{EQ mode check},
defined in Section~\ref{section:modeChecking}, to determine which
columns may need to be compared for equality when the policy is
audited.


\subsection{Adversary model and Security Goals}
\label{sec:log-equiv}

\vspace{5pt}\noindent \textbf{Assumptions and threat model.}
In our threat model, \cl is trusted but \CS is an \emph{honest but
 curious} adversary with the following capability: \CS can run any
 polynomial time algorithm on the stored (encrypted) log, including
 the audit over any policy.
We assume that \cl
generates keys and encrypts the log with our encryption schemes before
uploading it to \CS. Audit runs on the \CS infrastructure
but (by design) it does not perform decryption. Hence, \CS never sees
plaintext data or the keys, but \CS can glean some information about
the log, e.g., the order of two fields or the equality of two
fields. The output of audit may contain encrypted values indicating
policy violations, but these values are decrypted only at \cl.

We assume that privacy policies are known to the adversary.  This
assumption may not be true for an organization's internal policies,
but relaxing this assumption only simplifies our formal development.
To audit over logs encrypted with \dt, any constants appearing in the
policy (like ``Doctor'' in Policy~1 of
Section~\ref{section:introduction}) must be encrypted before the audit
process starts, so \CS can recover the association between ciphertext
and plaintexts of constants that appear in the (publicly known)
privacy policy.  Similarly, in \kh, the hashes of constants in
policies must be revealed to the adversary. 
in a set

\vspace{5pt}\noindent\textbf{Security and functionality goals.}
(Confidentiality) Our primary goal is to protect the confidentiality
of the log's content, despite any compromise of \CS, including its
infrastructure, employees, and the audit process running on
it.
(Expressiveness) Our system should be expressive enough to represent
and audit privacy policies derived from real legislation. In our
evaluation, we work with privacy rules derived from HIPAA and GLBA.

\vspace{5pt}\noindent\textbf{Log equivalence.} Central to the
definition of the end-to-end security property that we prove of our
\dt and \kh is the notion of log equivalence. It characterizes
what information about the database \emph{remains confidential}
despite a complete compromise of \CS.  Our security definition states
that the adversary can only learn that the log belongs to a stipulated
equivalence class of logs. The coarser our equivalence, the stronger
our security theorem.

For semantically secure encryption, we could say that two logs are
equivalent if they are the same length. When the encryption permits
join, selection, comparison and dispaced comparison queries, this
definition is too strong. For example, the attacker must be allowed to
learn that two constants on the log (e.g., Doctor and Nurse) are not
equal if they lie in different columns that the attacker can try to
join. Hence, we need a refined notion of log equivalence, which we
formalize in Section~\ref{section:security:det}.

%% file: spec.tex
We review the logic that we use to represent privacy policies and give
a formal definition of logs (databases). These definitions are later
used in the definition and analysis of our encryption schemes and
the \ereduce audit algorithm.

\vspace{5pt}\noindent\textbf{Policy logic.}
We use
the guarded-fragment of first-order
logic introduced in~\cite{ANB98} to represent privacy policies.
The syntax of the logic is shown
in Figure~\ref{fig:spec}. Policies or formulas are denoted \policy.
\begin{figure}[t]
\centering
\(
\begin{array}{cccl}
 \pred{Atoms} & \mathcal{P} & ::= &  \pred{p}(t_1,\ldots,t_n)\mid \pred{timeOrder}(t_1, d_1, t_2, d_2) \mid \\
 & & & t_1 = t_2 \\
 \pred{Guard} & g & ::= & \mathcal{P}\mid\top\mid\bot\mid g_1\wedge g_2\mid g_1\vee g_2\mid\exists x.g\\
 \pred{Formula} & \policy & ::= & \mathcal{P}\mid\top\mid\bot\mid\policy_1\wedge\policy_2\mid\policy_1\vee\policy_2\mid\\
 & & & \forall\vec{x}.(g\rightarrow\policy)\mid\exists\vec{x}.(g\wedge\policy)
\end{array}
\)
\vspace*{-10pt}
\caption{Policy specification logic syntax\label{fig:spec}}
\end{figure}
Terms $t$ are either constants $c,d$ drawn from a domain \cD or
variables $x$ drawn from a set \var. (Function symbols are
disallowed.) $\vec{t}$ denotes a list of terms. The basic building
block of formulas is \emph{atoms}, which represent relations between
terms. We allow three kinds of atoms. First,
$\pred{p}(t_1,\ldots,t_n)$ represents a relation which is established
through a table named \pred{p} in the audit log. The symbol \pred{p}
is called a predicate (or, interchangeably, a table). The set of all
predicate symbols is denoted by $\mathbb{P}$. An arity function
$\alpha: \mathbb{P} \rightarrow \mathbb{N}$ specifies how many
arguments each predicate takes (i.e., how many columns each table
has).  Second, for numerical terms, we allow comparison after
displacement with constants, written $\pred{timeOrder}(t_1, d_1, t_2,
d_2)$. This relation means that $t_1 + d_1 \leq t_2 + d_2$. Here,
$d_1, d_2$ must be constants.  Third, we allow term equality, written
$t_1 = t_2$.  Although we restrict atoms of the logic to these three
categories only, the resulting fragment is still very expressive. All
the HIPAA- and GLBA-based policies tested in prior work~\cite{GJD11}
and all but one clause of the entire HIPAA and GLBA privacy rules
formalized by DeYoung \etal~\cite{DGJKD10} lie within this fragment.

Formulas or policies, denoted \policy, contain standard logical
connectives $\top$ (``true''), $\bot$ (``false''), $\wedge$ (``and''),
$\vee$ (``or''), $\forall x$ (``for every $x$'') and $\exists x$
(``for some x''). Saliently, the form of quantifiers $\forall x$ and
$\exists x$ is restricted: Each quantifier must include
a \emph{guard}, $g$. As shown in~\cite{GJD11}, this restriction,
together with the mode check described in
Section~\ref{section:modeChecking}, ensures that audit terminates (in
general, the domain \cD may be infinite).
Intuitively, one may think of a policy $\policy$ as enforcing a
constraint on the predicates it mentions, i.e., on the tables of the
log. A guard $g$ may be thought of as a query on the log (indeed, the
syntax of guards generalizes Datalog, a well-known database query
language). The policy $\forall \vec{x}. (g \rightarrow \varphi)$ may
be read as ``for every result $\vec{x}$ of the query $g$, the
constraint $\varphi$ must hold.'' Dually,
$\exists \vec{x}. (g \wedge \varphi)$ may be read as ``some result
$\vec{x}$ of the query $g$ must satisfy the constraint $\varphi$.''

\vspace{5mm}
\noindent 
\textbf{Example 1.} 
Consider the following policy, based on \S6802(a) of the GLBA privacy
law: 

\noindent 
\[
\begin{array}{l}
\forall p_1, p_2, m, q, a, t.\quad(\pred{send}(p_1, p_2, m, t)\wedge\\ 
\pred{tagged}(m, q, a)\wedge \pred{activeRole}(p_1,\mathit{institution})\wedge\\
\pred{notAffiliateOf}(p_2, p_1, t)\wedge\pred{customerOf}(q, p_1, t)\wedge
\pred{attr}(a, \mathit{npi}))\\
\hspace*{1em}\rightarrow\bigg((\exists t_1, m_1.\pred{send}(p_1, q, m_1, t_1)\wedge\pred{timeOrder}(t_1,0,t,0)\wedge\\
\hspace*{2em}
\pred{timeOrder}(t, 0, t_1, 30)\wedge
\pred{discNotice}(m_1, p_1, p_2, q, a,t))\\
\hspace*{10em}
\bigvee\\
\hspace*{1em}(\exists t_2, m_2.\pred{send}(p_1, q, m_2, t_2)\wedge\pred{timeOrder}(t,0,t_2,0)\wedge\\
\hspace*{2em}\pred{timeOrder}(t_2, 0, t, 30)\wedge\pred{discNotice}(m_2, p_1, p_2, q, a,t))\bigg)
\end{array}
\]
The policy states that principal $p_1$ can \emph{send} a message $m$
to principal $p_2$ at time $t$ where the message $m$ contains
principal $q$'s attribute $a$ (\eg, account number) and (i) $p_1$ is
in the role of a financial institution, (ii) $p_2$ is \emph{not} a
third-party affiliate of $p_1$ at time $t$, (iii) $q$ is a customer of
$p_1$ at time $t$, (iv) the attribute $a$ is non-public personal
information ($npi$, \eg, a social security number) only if any one of
the two conditions separated by $\vee$ holds. The first condition says
that the institution has already sent a notification of this
disclosure in the past 30 days to the customer $q$ (\ie, $0\leq
(t-t_1)\leq 30$).  The second condition says that the institution will
send a notification of this disclosure within the next 30 days (\ie,
$0\leq (t_2-t)\leq 30$).

\vspace{5pt}\noindent \textbf{Logs and schemas.}
An audit log or log, denoted \cL, is a database with a given schema. A
schema \cS is a set of pairs of the form
$\langle \pred{tableName}, \pred{columnNames}\rangle$
where \pred{columnNames} is an ordered list of all the column names in
the table (predicate) \pred{tableName}. A schema \cS corresponds to a
policy \policy if \cS contains all predicates mentioned in the
policies \policy, and the number of columns in predicate \pred{p} is
$\alpha(\pred{p})$.

Semantically, we may view a log \cL as a function that given as argument
a variable-free atom $\pred{p}(\vec{t})$ returns either $\top$ (the
entry $\vec{t}$ exists in table \pred{p} in \cL) or $\bot$ (the entry
does not exist). To model the possibility that a log table may be
incomplete, we allow for a third possible response \pred{uu}
(unknown). In our implementation, the difference between \pred{uu} and
$\bot$ arises from an additional bit on the table \pred{p} indicating
whether or not the table may be extended in future.  Formally, we say
that log $\cL_1$ extends log $\cL_2$, written $\cL_1 \geq \cL_2$ when
for every $\pred{p}$ and $\vec{t}$, if
$\cL_2(\pred{p}(\vec{t})) \not= \pred{uu}$, then
$\cL_1(\pred{p}(\vec{t})) = \cL_2(\pred{p}(\vec{t}))$. Thus, the
extended log $\cL_1$ may determinize some unknown entries from
$\cL_2$, but cannot change existing entries in $\cL_2$.

Our logic uses standard semantics of first-order logic, treating logs
as models. The semantics, written $\cL \models \varphi$, take into
account the possibility of unknown relations; we refer the reader
to~\cite{GJD11} for details (these details are not important for
understanding this paper). Intuitively, if $\cL \models \varphi$, then
the policy $\varphi$ is satisfied on the log $\cL$; if
$\cL \not \models \varphi$, then the policy is violated; and if
neither holds then the log does not have enough information to
determine whether or not the policy has been violated.

\vspace{5mm}
\noindent \textbf{Example 2.} 
The policy in Example 1 can be checked for violations on a log whose
schema contains tables $\pred{send}$, $\pred{tagged}$,
$\pred{activeRole}$, $\pred{notAffiliateOf}$, $\pred{customerOf}$,
$\pred{attr}$ and $\pred{discNotice}$ with 4, 3, 2, 3, 3, 2 and 6
columns respectively. In this audit, values in several columns may
have to be compared for equality. For example, the values in the first
columns of tables $\pred{send}$ and $\pred{activeRole}$ must be
compared because, in the policy, they contain the same variable
$p_1$. Similarly, timestamps must be compared after displacement with
constants $0$ and $30$. The log encryption schemes we define next
support these operations.

%% file: function.tex

We present our two log encryption schemes, \dt and \kh in
Section~\ref{section:det} and Section~\ref{section:kh}
respectively. Both schemes use (as a black-box) a new sub-scheme,
\DOPE, for comparing timestamps after displacement, which we present
in Section~\ref{section:dope}.

\subsection{Preliminaries}
We introduce common constructs used through out the rest of this
section.

\vspace{5pt}
\noindent\textbf{Equality scheme.}
To support policy audit, we determine, through
 a static analysis of the policies to be audited, which pairs of
 columns in the log schema may have be tested for equality or joined.
We defer the details of this
policy analysis to Section~\ref{section:modeChecking}. For now, we
just assume that the result of this analysis is available. This
result, called an \emph{equality scheme}, denoted \escheme, is a set of
pairs of the form $\langle \pred{p}_1.\pred{a_1},
\pred{p}_2.\pred{a_2}\rangle$. The key property of \escheme is that
if, during audit, column \pred{a_1} of table \pred{p_1} is tested for
equality against column \pred{a_2} of table \pred{p_2}, then $\langle
\pred{p}_1.\pred{a_1}, \pred{p}_2.\pred{a_2}\rangle \in
\escheme$.

\vspace{5pt}
\noindent\textbf{Policy constants.}
Policies may contain constants. For instance, the policy of Example 1
contains the constants $npi$, $institution$, $0$ and $30$.
Before running our audit algorithm over encrypted logs, a new
version of the policy containing these constants in either encrypted
(for \dt) or keyed hash (for \kh) form must be created.
Consequently, the
adversary, who observes the audit and knows the plaintext policy, can
learn the encryption or hash of these constants. Hence, these
constants play an important role in our security definitions. The set
of all these policy constants is denoted $C$.

\vspace{5pt}
\noindent\textbf{Displacement constants.}
Constants which feature in the 2nd and 4th argument positions
of the predicate $\pred{timeOrder}()$ 
play a significant
role in construction of the \DOPE encoding and our security
definition. These constants are called \emph{displacements}, denoted
$D$. For instance, in Example 1, $D = \{0, 30\}$. For any policy, $D
\subseteq C$. 

\vspace{5pt}
\noindent\textbf{Encrypting timestamps.}
We assume (conservatively) that all timestamps in the plaintext log
may be compared to each other, so all timestamps are encrypted
(in \dt) or hashed (in \kh) with the same key $K_\pred{time}$. This
key is also used to protect values in the \DOPE sub-scheme. The
assumption of all timestamps may be compared with each other, can be
restricted substantially (for both schemes) if the audit policy is
fixed ahead of time.

\subsection{Eunomia$^\mathbf{DET}$}\label{section:det}

The log encryption scheme \dt encrypts each cell individually using
deterministic encryption. All cells in a column are encrypted with the
same key. Importantly, if cells in two columns may be compared during
audit (as determined by the equality scheme \escheme), then the two
columns also share the same key. Hence, cells can be tested for
equality simply by comparing their ciphertexts. To allow timestamp
comparison after displacement, the encrypted log is paired with
a \DOPE encoding of timestamps that we explain later. Note that it is
possible to replace deterministic encryption with a cryptographically
secure keyed hash and a semantically secure ciphertext to achieve the
same functionality (the keyed hash value could be used to check for
equality). However, this design incurs higher space overhead than our
design with deterministic encryption.

Technically, \dt contains the following three algorithms:
$\pred{KeyGen}^\pred{DET}(1^\param,\cS,\escheme)$,
$\pred{EncryptLog}^\pred{DET}(\cL,\cS,\keys)$, and
$\pred{EncryptPolicyConstants}^\pred{DET}(\policy,\keys)$.

\vspace{5pt}
\noindent \textbf{Key generation.}  The probabilistic algorithm
$\pred{KeyGen}^\pred{DET}(\cdot,\\\cdot,\cdot)$ takes as input the
security parameter \param, the plaintext log schema \cS, and an equality
scheme \escheme. It returns a \emph{key set} \keys. The \emph{key set}
\keys is a set of triples of the form $\langle \pred{p}, \pred{a},
k\rangle$. The triple means that all cells in column $\pred{a}$ of
table $\pred{p}$ must be encrypted (deterministically) with key
$k$. The constraints on \keys are that (a) if \pred{p}.\pred{a}
contains timestamps, then $k = K_\pred{time}$, and (b) if $\langle
\pred{p}_1.\pred{a_1}, \pred{p}_2.\pred{a_2}\rangle \in \escheme$,
$\langle \pred{p_1}, \pred{a_1}, k_1 \rangle \in \keys$ and $\langle
\pred{p_2}, \pred{a_2}, k_2 \rangle \in \keys$, then $k_1 = k_2$.


\vspace{5pt}
\noindent
\textbf{Encrypting the log.} The algorithm
$\pred{EncryptLog}^\pred{DET}(\cdot,\cdot,\cdot)$ takes as input a
plaintext log \cL, its schema \cS, and the key set \keys generated by
$\pred{KeyGen}()$. It returns a pair $e\cL=\langle e\db, e\cT\rangle$
where, $e\db$ is the cell-wise encryption of \cL with appropriate keys
from \keys and $e\cT$ is the \DOPE encoding.



\vspace{5pt}
\noindent
\textbf{Encrypting constants in the policy.}  To audit over logs
encrypted with \dt, constants in the policy must be encrypted too
(else, we cannot check whether or not an atom mentioning the constant
appears in the encrypted log). The algorithm
$\pred{EncryptPolicyConstants}^\pred{DET}(\cdot,\cdot)$ takes as input
a plaintext policy \policy, and a key set \keys, and returns a policy
$\policy'$ in which constants have been encrypted with appropriate
keys. The function works as follows: If, in \policy, the constant $c$
appears in the $i$th position of predicate \pred{p}, then in
$\policy'$, the $i$th position of \pred{p} is $c$ deterministically
encrypted with the key of the $i$th column of \pred{p} (as obtained
from \keys). Other than this, $\policy$ and $\policy'$ are identical.


\vspace{5mm}
\noindent
\textbf{Remarks.}  The process of audit on a log encrypted with \dt
requires no cryptographic operations. Compared to an unencrypted log,
we only pay the overhead of having to compare longer ciphertexts and
some cost for looking up the \DOPE encoding to compare
timestamps. However, auditing for a policy that requires equality
tests beyond those prescribed by an equality scheme
\escheme is impossible on a log encrypted for \escheme. To do so, we
would have to re-encrypt parts of the log, which is a slow
operation. Our second log encryption scheme, \kh, represents a
different trade-off.




\subsection{Eunomia$^\mathbf{KH}$}\label{section:kh}

\kh relies on the adjustable keyed hash (AKH)
scheme~\cite{ADJJOIN,cryptDB} to support equality tests. We
review AKH and then describe how we build \kh on it.

Abstractly, AKH provides three functions:
 $\pred{Hash}(k,v)=\pred{P}\times\pred{DET}(k_\pred{master},v)\times
 k$ (\pred{P} is a point on an elliptic curve,
 $\pred{DET}(\cdot,\cdot)$ is the deterministic encryption function,
 and $k_\pred{master}$ is a master encryption key),
 $\pred{Token}(k_1,k_2)=k_2\times k_1^{-1}$ and
 $\pred{Adjust}(w,\Delta)=w\times\Delta$. $\pred{Hash}(k,v)$ returns a
 keyed hash of $v$ with key $k$ on a pre-determined elliptic curve
 with public parameters. $\pred{Token}(k_1,k_2)$ returns a token
 $\Delta_{k_1 \mapsto k_2}$, which allows \emph{transforming} hashes
 created with key $k_1$ to corresponding hashes created with
 $k_2$. The function $\pred{Adjust}(w,\Delta)$ performs this
 transformation: If $w = \pred{Hash}(k_1,v)$ and $\Delta
 = \Delta_{k_1 \mapsto k_2}$, then $\pred{Adjust}(w, \Delta)$ returns
 the same value as $\pred{Hash}(k_2,v)$. The AKH scheme allows the
 adversary to compare two values hashed with keys $k_1$ and $k_2$ for
 equality only when it knows either $\Delta_{k_1 \mapsto k_2}$ or
 $\Delta_{k_2 \mapsto k_1}$. Popa \etal prove this security property,
reducing it to the elliptic-curve decisional Diffie-Hellman
assumption~\cite{ADJJOIN}.

To encrypt a log in \kh, we generate two keys $k_h, k_e$ for each
column. These are called the hash key and the encryption key,
respectively. Each cell $v$ in the column is transformed into a pair
$\langle \pred{Hash}(k_h, v), \pred{Encrypt}(k_e, v) \rangle$. Here,
$\pred{Hash}(k_h,v)$ is the AKH hash of $v$ with key $k_h$ and
$\pred{Encrypt}(k_e, v)$ is a standard probabilistic encryption of $v$
with key $k_e$.%
\footnote{
The $\pred{Encrypt}(k_e, v)$ component of the ciphertext is returned
to the client \cl as part of the audit output. {\cl} then decrypts it
to obtain concrete policy violations.}
Columns do not share any keys. If audit on a policy requires testing
columns $\pred{t_1}.\pred{a_1}$ and $\pred{t_2}.\pred{a_2}$ for
equality and these columns have hash keys $k_{h_1}$ and $k_{h_2}$,
then the audit algorithm is given one of the tokens
$\Delta_{k_{h_1} \mapsto k_{h_2}}$ and $\Delta_{k_{h_2} \mapsto
k_{h_1}}$. The algorithm can then transform hashes to test for
equality. Each execution of the audit process can be given a different
set of tokens depending on the policy being audited and, hence, unlike
\dt, the same encrypted log supports audit over any policy. However,
equality testing is more expensive now as it invokes the
$\pred{Adjust}()$ function. This increases the runtime overhead of
audit.

Formally, \kh contains the following four algorithms:
$\pred{KeyGen}^\pred{KH}(1^\param,\cS)$,
$\pred{EncryptLog}^\pred{KH}(\cL,\cS,\keys)$,
$\pred{EncryptPolicyConsta}\\\pred{nts}^\pred{KH}(\policy,\keys)$, and
$\pred{GenerateToken}(\cS,\escheme, \keys)$.

\vspace{5mm}
\noindent\textbf{Key generation.} The probabilistic algorithm
$\pred{KeyGen}^\pred{KH}(\cdot,\cdot)$ takes as input a security
parameter and a log schema \cS and returns a key set \keys. \keys
contains tuples of the form $\langle \pred{p}, \pred{a}, k_h,
k_e\rangle$, meaning that column $\pred{p}.\pred{a}$ has hash key
$k_h$ and encryption key $k_e$. The only constraint is that if
$\pred{p}.\pred{a}$ contains timestamps, then $k_h = K_\pred{time}$.

\vspace{5mm}
\noindent\textbf{Encrypting the log.} The algorithm
$\pred{EncryptLog}^\pred{KH}(\cdot,\cdot,\cdot)$ takes as arguments a
plaintext log \cL, its schema \cS and a key set \keys. It returns a
pair $e\cL=\langle e\db, e\cT\rangle$ where, $e\db$ is the cell-wise
encryption of \cL with appropriate keys from \keys and $e\cT$ is
the \DOPE encoding. Because each cell maps to a pair, each table has
twice as many columns in $e\db$ as in~$\cL$.

\vspace{5mm}
\noindent\textbf{Encrypting policy constants.} To audit over \kh encrypted logs,
constants in the policy must be encrypted. The
algorithm $\pred{EncryptPolicyConstants}^\pred{KH}(\cdot,\cdot)$ takes
as input a plaintext policy \policy, and a key set \keys, and returns
a policy $\policy'$ in which constants have been encrypted with
appropriate keys taken from \keys: If constant $c$ appears in the
$i$th position of predicate \pred{p} in \policy and the hash and
encryption keys of the $i$th column of \pred{p} in \keys are $k_h$ and
$k_e$, respectively, then the constant $c$ is replaced by $\langle
\pred{Hash}(k_h,v), \pred{Encrypt}(k_e,v)\rangle$ in $\policy'$.

\vspace{5mm}
\noindent\textbf{Generating tokens.}
$\pred{GenerateToken}(\cdot,\cdot,\cdot)$ is used to generate tokens
that are given to the audit algorithm to enable it to test for
equality on the encrypted log. For each tuple $\langle
\pred{p}.\pred{a}_1, \pred{q}.\pred{a}_2\rangle$ in \escheme, the
algorithm $\pred{GenerateToken}(\cS,\escheme, \keys)$ returns the
tuple $\langle \pred{p}.\pred{a}_1, \pred{q}.\pred{a}_2,
\Delta_{k_1\mapsto k_2}\rangle$, where $\langle\pred{p}, \pred{a}_1,
k_1, \_\rangle\in\keys$ and $\langle\pred{q}, \pred{a}_2,
k_2,\_\rangle\in\keys$.

\vspace{5mm}
\noindent\textbf{Remark.}
From the perspective of confidentiality, the same amount of
information is revealed irrespective of whether the audit algorithm
(which may be compromised by the adversary) is given
$\Delta_{k_{h_1} \mapsto k_{h_2}}$ or $\Delta_{k_{h_2} \mapsto
k_{h_1}}$, because each token can be computed from the other. However,
the actual token used for comparison by the audit algorithm can have a
significant impact on its performance. Consider Policy 1 from
Section~\ref{section:introduction}, which stipulates that each name
appearing in table $T_1$ appear in $T_2$ with the role Doctor.  The
audit process will iterate over the names in $T_1$ and look up those
names in $T_2$. Consequently, for performance, it makes sense to index
the hashes of the names in $T_2$ and for the audit algorithm to use
the token $\Delta_{k_1 \mapsto k_2}$, where $k_1$ and $k_2$ are the
hash keys of names in $T_1$ and $T_2$, respectively. If, instead, the
algorithm uses $\Delta_{k_2 \mapsto k_1}$, then indexing is
ineffective and performance suffers. The bottom line is
that \emph{directionality} of information flow during equality testing
matters for \kh. Our policy analysis, which determines the columns
that may be tested for equality during audit
(Section~\ref{section:modeChecking}) takes this directionality into
account. The equality scheme \escheme returned by this analysis is
directional (even though the use of \escheme in \dt ignored
this directionality): if $\langle \pred{p}_1.\pred{a_1},
\pred{p}_2.\pred{a_2}\rangle \in \escheme$, and
$\pred{p}_1.\pred{a_1}$ and $\pred{p}_2.\pred{a_2}$ have hash keys
$k_1$ and $k_2$, then the audit algorithm uses the token
$\Delta_{k_1 \mapsto k_2}$, not $\Delta_{k_2 \mapsto k_1}$.

\subsection{Mutable Order Preserving Encoding with Displacements (mOPED)}
\label{section:dope}

We now discuss the \DOPE scheme which produces a data structure,
$e\cT$, that allows computation of the boolean value $t_1 + d_1 \leq
t_2 + d_2$ on the cloud, given only $\pred{Enc}(t_1)$,
$\pred{Enc}(t_2)$, $\pred{Enc}(d_1)$ and $\pred{Enc}(d_2)$. Here,
$\pred{Enc}(t)$ denotes the deterministic encryption of $t$ (in the
case of \dt) or the AHK hash of $t$ (in the case of \kh) with the
fixed key $K_\pred{time}$. The scheme \DOPE extends a prior
scheme \MOPE~\cite{mope}, which is a special case $d_1 = d_2 = 0$ of
our scheme.

Consider first the simple case where the log $\cL$ and the policy
\policy are fixed. This means that the set $T$ of values of the form
$t+d$ that the audit process may compare to each other is also fixed
and finite (because $t$ is a timestamp on the finite log $\cL$ and $d
\in D$ is a displacement occurring in the finite policy \policy). Suppose
that the set $T$ has size $N$ (note that $N \in O(|D| \cdot
|\cL|$). Then, the client can store on the cloud a map
$e\cT: \pred{EncTimeStamp} \times \pred{EncD} \rightarrow
\{1,\ldots,N\}$, which maps each encrypted timestamp
$\pred{Enc}(t)$ and each encrypted displacement
$\pred{Enc}(d)$ to the relative order of $t+d$ among the
elements of $T$. To compute $t_1 + d_1 \leq t_2 + d_2$, the audit
process can instead compute $e\cT(\pred{Enc}(t_1),
\pred{Enc}(d_1)) \leq e\cT(\pred{Enc}(t_2),
\pred{Enc}(d_2))$. The map $e\cT$ can be represented in many
different ways. In our implementation, we use nested hash tables,
where the outer table maps $\pred{Enc}(t)$ to an inner hash table and
the inner table maps $\pred{Enc}(d)$ to the relative order of
$t+d$. For audit applications where the log and policy are fixed
upfront, this simple data structure $e\cT$ suffices. 

The scheme \DOPE is more general and allows the client to
incrementally update $e\cT$ on the cloud. This is relevant when either
the policy or the log changes often. A single addition or deletion of
$t$ or $d$ can cause the map $e\cT$ to change for potentially all
other elements and, hence, a naive implementation of $e\cT$ may incur
cost linear in the current size of $T$ for single updates. Popa \etal
show how this cost can be made logarithmic by interactively
maintaining a balanced binary search tree over encrypted values
$\pred{Enc}(t)$ and using \emph{paths} in this search tree as the
co-domain of $e\cT$. We extend this approach by maintaining a binary
search tree over pairs $(\pred{Enc}(t),\pred{Enc}(d))$, where the
search order reflects the natural order over $t+d$. Since the cloud
never sees plaintext data, the update of this binary search tree and
the map $e\cT$ \emph{must} be interactive with the client. We omit the
details of this interactive update and refer the reader to~\cite{mope}
for details. As the cloud may be compromised, the security property we
prove of \DOPE (Section~\ref{section:security:ope}) holds despite the
adversary observing every interaction with the client. We note that an
audit algorithm never updates $e\cT$, so its execution remains
non-interactive.

%% file: security2.tex
We now prove that our schemes \dt, \kh and \DOPE are secure.  We start
with \DOPE, because \dt and \kh rely on it.



\subsection{Security of mOPED}
\label{section:security:ope}

We formalize the security of \DOPE as an indistinguishability game in
which the adversary provides two sequences of timestamps and a set of
displacements $D$, then observes the client and server construct
the \DOPE data structure $e\cT$ on one of these sequences chosen
randomly and then tries to guess which sequence it is. We call this
game \DOCPA (indistinguishability under chosen distances with
displacement attack). This definition is directly based on the \OCPA
(indistinguishability under ordered chosen plaintext attack)
definition by Boldyreva~\etal~\cite{BCLO} and the \LOR security
definition by Pandey and Rouselakis~\cite{PR12}. Because $e\cT$
intentionally reveals the relative order of all timestamps after
displacement with constants in $D$, we need to impose a constraint on
the two sequences chosen by the adversary. Let $\vec{u}[i]$ denote the
$i$th element of the sequence $\vec{u}$. We say that two sequences of
timestamps $\vec{u}$ and $\vec{v}$ are equal up to distances with
displacements $D$, written $\pred{EDD}(\vec{u}, \vec{v}, D)$ iff
$|\vec{u}| = |\vec{v}|$ and $\forall d,d' \in D, i, j. ~(\vec{u}[i] +
d \geq \vec{u}[j] + d') \Leftrightarrow (\vec{v}[i] +
d \geq \vec{v}[j] + d')$. We describe here the \DOCPA game and the
security proof for \DOPE with deterministic encryption; the case
of \DOPE with AKH hashes is similar.

\vspace{5mm}
\noindent
\textbf{\DOCPA game.}
The \DOCPA security game between a client or challenger \cl (\ie,
owner of the audit log) and an adaptive, probabilistic polynomial time
(ppt) adversary \adv for the security parameter \param proceeds as
follows:

\begin{packedenumerate}
\item \cl generates a secret key $K_\pred{time}$ using the probabilistic key generation algorithm \keygen.
 $K_\pred{time} \overset{\$}{\leftarrow} \keygen(1^\param)$.
\item \cl chooses a random bit $b$. $b\overset{\$}{\leftarrow}\{0, 1\}$.
\item \cl creates an empty $e\cT$ 
 on the cloud.

\item \adv chooses a set of
  distances $D =\{d_1,\ldots,d_n\}$ and sends it to \cl.

\item \cl and \adv engage in a polynomial of \param number of rounds
 of interactions. In each round $j$:

\begin{packedenumerate}
\item \adv selects two values $v^0_j$ and $v^1_j$ and sends them to \cl.

\item \cl deterministically encrypts the following $n+1$ values
  $v^b_j$, $v^b_j + d_1, v^b_j + d_2, \ldots, v^b_j + d_n$ using
  $K_\pred{time}$.

\item \cl interacts with the cloud to insert
  $\deter(K_\pred{time}, v^b_j)$ and $\{\deter(K_\pred{time},
  v^b_j+d_i)\}_{i=1}^n$ into $e\cT$. 
The adversary
  observes this interaction and the cloud's complete state, but
  not \cl's local computation.
\end{packedenumerate}

\item \adv outputs his guess $b'$ of $b$.
\end{packedenumerate}

\adv wins the \DOCPA security game iff:
\begin{packedenumerate}\setlength{\itemsep}{0em}
\item \adv guesses  $b$ correctly (\ie, $b=b'$);
\item $\pred{EDD}([v_0^0,\ldots,v_m^0],[v_0^1,\ldots,v_m^1],D)$ holds, where $m$ is the number of rounds played in the game.
\end{packedenumerate}

Let $\mathsf{win}^{\adv,\kappa}$ be a random variable which is 1 if
the \adv wins and 0 if \adv loses. Recall that a function
$f:\mathbb{N}\rightarrow\mathbb{R}$ is \emph{negligible} with respect
to its argument \param, if for every $c\in\mathbb{N}$ there exists
another integer $K$ such that for all $\param > K$, $f(\param) <
x^{-c}$. We write $\pred{negl}(\param)$ to denote some neglible
function of $\kappa$.


\begin{theorem}[Security of \DOPE with deterministic encryption]
\label{theorem:dopeSecurity}
Assuming that deterministic encryption is a pseudorandom function,
our \DOPE scheme is \DOCPA secure, \ie,
$\mathsf{Pr}[\mathsf{win}^{\adv, \kappa}=1]\leq \frac{1}{2}
+ \mathsf{negl}(\param)$ where the probability is taken over the
random coins used by \adv as well as the random coins used in choosing
the key and the
bit $b$.
\end{theorem}
\begin{proof}
By a hybrid argument. We augment a similar proof of security for
the \MOPE scheme~\cite{mope} to also take displacements $D$ into
account.
\end{proof}

\noindent\textbf{Security of \DOPE with AKH hash.}
The security game for \DOPE with AKH hashes is very similar
to \DOCPA. We replace $\pred{DET}(\cdot,\cdot)$ with
$\pred{Hash}(\cdot,\cdot)$ in the game. The proof is in the standard
model and reduces to the security of AKH~\cite[Definition 4]{ADJJOIN}
and finally to the elliptic-curve decisional Diffie-Hellman (ECDDH)
assumption.


\subsection{Security of Eunomia$^\mathbf{DET}$}
\label{section:security:det}

We prove security for \dt, formalized as an indistinguishability
game. We first define a notion of log equivalence that characterizes
the confidentiality achieved by \dt (and, as we explain later,
by \kh). This notion is a central contribution of our work. The
security theorem in this section shows that by looking at the \dt
encryption of a log, a PPT adversary can learn only that the log
belongs to its equivalence class (with non-negligible
probability). Hence, the equivalence class of the log represents the
uncertainty of the adversary about the log's contents and, therefore,
characterizes what confidentiality the scheme provides.

\begin{definition}[Plaintext log equivalence]\label{def:plainLogEquivalence}
Given two plaintext audit logs $\cL_1$ and $\cL_2$, an equality scheme
\escheme, a set of constants $C$ and a set of displacements $D\subseteq C$, $\cL_1$ and $\cL_2$ are equivalent, denoted by
$\cL_1\equiv_{(\escheme, C,D)}\cL_2$, if and only if all of the
followings hold:
\begin{packedenumerate}\setlength{\itemsep}{-0.5em}
\item $\cL_1$ and $\cL_2$ have the same schema and tables of the same
  name in $\cL_1$ and $\cL_2$ have the same number of records (rows).
\item For each equivalence class of columns defined by \escheme,
there is a bijection from values of $\cL_1$ to values of $\cL_2$.  (By
equivalence class of columns defined by \escheme, we mean an
equivalence class of columns defined by the reflexive, symmetric,
transitive closure of \escheme.)  For a table $\pred{t}$ and a column
$\pred{a}$, let $\cM_{\pred{t},\pred{a}}$ denote the bijection
corresponding to the equivalence class of \escheme in which
$(\pred{t}, \pred{a})$ lies. Let $v$ be the value in some row $i$ of
the table $\pred{t}$, column $\pred{a}$ in $\cL_1$. Then,
\begin{packedenumerate}
\item The value in the $i$th row of table $\pred{t}$, column
  $\pred{a}$ in $\cL_2$ is $\cM_{\pred{t},\pred{a}}(v)$.
\item If $v \in C$, then $\cM_{\pred{t},\pred{a}}(c) = c$.
\item $|v| = |\cM_{\pred{t},\pred{a}}(v)|$.
\end{packedenumerate}
\item Let $\pred{timeStamps}(\cL_1)$ be the sequence of timestamps in
  $\cL_1$ obtained by traversing the tables of $\cL_1$ in any order
  and the timestamps within each table in row order. Let
  $\pred{timeStamps}(\cL_2)$ be the timestamps in $\cL_2$ obtained
  similarly, traversing tables in the same order. Then,
  $\pred{EDD}(\pred{timeStamps}(\cL_1), \pred{timeStamps}(\cL_2), D)$
  holds.
\end{packedenumerate}
\end{definition}

We now intuitively explain the requirements of two plaintext log $\cL_1$ and $\cL_2$
to be equivalent
according to our definition.
The first requirement of our definition requires both $\cL_1$ and $\cL_2$ to
have the same schema, same table names, and same number of records. The second
requirement states that for each equivalence class of columns defined by the
reflexive, symmetric, transitive closure of \escheme, there exists a bijection
mapping $\cM$ from the plaintext values appearing in those columns in $\cL_1$
to the plaintext values appearing in those columns in $\cL_2$ such that:
(a) All constant values $c$ appearing in those columns (\ie, $c\in C$) are mapped to each other;
(b) If there is a mapping between two values $p_1$ to $p_2$ according to \cM,
then $p_1$ and $p_2$ are of same length;
(c) If we take any arbitrary row $i$ in any arbitrary column $j$ (where $j$ belongs to the table $T$ and
the equivalence class in question) in $\cL_1$ and assume $\cL_1.T[i][j] = v_1$, then
if we apply $\cM$ over $v_1$, we will get the value of $\cL_2.T[i][j]$, and vice versa.
The final requirement demands that if we take
any two arbitrary displacements $d_1, d_2\in D$,
any two arbitrary timestamps $t_i^1, t_j^1$
from $\cL_1$, and the values of the  corresponding  cells in  $\cL_2$ are
$t_i^2, t_j^2$, then the following holds: $t_i^1+d_1\leq t_j^1 + d_2\Leftrightarrow t_i^2+d_1\leq t_j^2 + d_2$.


We now define the \CPLA game (Chosen Plaintext Log
Attack), which defines what it means for
two logs to be indistinguishable to an adversary \adv
under \dt.

\noindent\textbf{\CPLA game.}
The \CPLA game is played between a client or challenger \cl and an
adversary
\adv for all large enough security parameters \param.\vspace*{-2pt}
\begin{enumerate}\setlength{\itemsep}{-2pt}
\item \adv picks a log schema \cS, the sets $C$, $D$ and an equality
  scheme \escheme and gives these to \cl.
\item \cl probabilistically generates a set of secret keys \keys based
  on the sufficiently large security parameter \param, the log
  schema \cS, and the equality scheme \escheme.
  $\keys\overset{\$}{\leftarrow} \pred{KeyGen}^\pred{DET}(1^\param,
  \cS, \escheme)$.
\item \cl randomly  selects a bit $b$.
$b\overset{\$}{\leftarrow}\{0,1\}$.
\item \adv chooses two plaintext audit logs $\cL_0$ and $\cL_1$ such
  that both $\cL_0, \cL_1$ have schema \cS, $\cL_0\equiv_{(\escheme,
    C,D)}\cL_1$, $\cL_0\neq \cL_1$, and sends $\cL_0$, $\cL_1$ to \cl.
\item Following the scheme \dt, \cl deterministically encrypts $\cL_b$
  according to the key set \keys to obtain the encrypted audit log
  $e\dbd{b}$. It then constructs the \DOPE data structure
  $e\cT_b$. \adv may observe the construction of the \DOPE
  data-structure $e\cT_b$ passively. \cl then sends $\langle e\dbd{b},
  e\cT_b\rangle$ to \adv.

($\langle e\dbd{b},
  e\cT_b\rangle\leftarrow\pred{EncryptLog}^\pred{DET}(\cL_b, \cS, \keys)$.)
\item For any constant $c \in C$, if $c$ appears in table $\pred{t}$,
  column $\pred{a}$ of $\cL_b$, then \cl gives \adv the encryption of $c$ with
  the encryption key of column $\pred{a}$.
\item \adv runs a probabilistic algorithm that may invoke the
  encryption oracle on keys from \keys but never asks for the
  encryption of any value in $\cL_0$ or $\cL_1$.
\item \adv outputs its guess $b'$ of $b$.
\end{enumerate}

\adv wins the 
\CPLA game \emph{iff} $b=b'$.  Let
the random variable $\pred{win}^{\adv}_{\pred{DET}}$ be 1 if the \adv
wins and 0 otherwise.


\begin{theorem}[Security of \dt]\label{thm:security:det}
If deterministic encryption is a pseudorandom function,  \dt
is \CPLA secure, i.e., for any ppt adversary \adv and sufficiently
large $\param$,
$\mathsf{Pr}[\pred{win}^{\adv}_{\pred{DET}}=1] \leq \frac{1}{2} +
\pred{negl}(\param)$ where the probability is taken over the random coins
used by \adv as well as the random coins used in choosing keys and the
random bit $b$.
\end{theorem}
\begin{proof}
By
hybrid argument. We successively replace uses of deterministic
encryption with a random oracle. If the \adv can distinguish two
consecutive hybrids with non-negligible probability, it can also
distinguish a random oracle from a pseudorandom function, which is a
contradiction.
\end{proof}

Intuitively, this theorem says that any adversary cannot distinguish
two equivalent logs if they are encrypted with \dt, except with
negligible probability.

\subsection{Security of Eunomia$^\mathbf{KH}$}
\label{section:security:kh}

We now define and prove security for the log encryption
scheme \kh. The security game, \CPLAKH, is similar to that for \dt and
uses the same notion of log equivalence. The proof of security for \kh
reduces to the ECDDH assumption.


\noindent\textbf{\CPLAKH game.}
The \CPLAKH game is played between a challenger \cl and an adversary
\adv for all large enough security parameters \param. It is very similar to
\CPLA security game but has the following differences. All the encryption is
done using the \kh approach. Additionally, there is one more step after step 5 where
the \cl generates the token list $\vec{\Delta}$ according to \escheme and send its to the
\adv. We do not show the \CPLAKH game due to space constraint.
\adv wins the
\CPLAKH game \emph{iff} $b=b'$.  Let
the random variable $\pred{win}^{\adv}_{\pred{KH}}$ be 1 if the \adv
wins and 0 otherwise.


\begin{theorem}[Security of \kh]\label{thm:security:kh}
If the ECDDH assumption holds and the encryption scheme used in \kh is
IND-CPA secure, then \kh is \CPLAKH secure, i.e., for any ppt
adversary \adv and sufficiently large $\param$, the following holds:
$\mathsf{Pr}[\pred{win}^{\adv}_{\pred{KH}}=1] \leq \frac{1}{2} +
\pred{negl}(\param)$, where the probability is taken over the random coins
used by \adv as well as the random coins used in choosing keys and the
random bit $b$.
\end{theorem}

\begin{proof}
By
hybrid argument, we reduce to the IND-CPA security of encryption
and the security of AKH~\cite[Definition 4]{ADJJOIN}. The latter
relies on the ECDDH assumption.
\end{proof}

\textbf{Generalizing security definitions. }\CPLA and \CPLAKH security definitions
are presented as a  game in which
the adversary and the challenger interact for a single round. However, it is possible to lift our current
security definitions for polynomial round of interactions between the adversary and challenger. To achieve this,
we need the requirement that for two consecutive rounds ($i$ and $i+1$) of interactions where $\cL^i_0$, $\cL^i_1$
are adversary-chosen plaintext equivalent logs at step $i$ and $\cL^{i+1}_0$, $\cL^{i+1}_1$ are adversary-chosen
plaintext equivalent logs at step $i+1$, $\cL^{i+1}_0 \geq \cL^{i}_0$ and $\cL^{i+1}_1 \geq \cL^{i}_1$.

\textbf{Information leakage. }Our security definitions
state that the only information leaked by our scheme is which log equivalence class a particular plaintext
audit log belongs to. Note that, a log's membership to an equivalence class is a symbolic representation of
all information that may leak by $\framework$. \emph{Precisely, by analyzing the log equivalence definition} (Definition~\ref{def:plainLogEquivalence})
\emph{it is possible to infer all information that may leak by $\framework$}. For instance, according to Definition~\ref{def:plainLogEquivalence},
two equivalent plaintext logs have the same frequency distribution in any column of the log.
However, two logs in different equivalence class may have different frequency distribution in any column.
Hence, $\framework$ leaks the frequency distribution in any column.

%% file: algorithm.tex
We now present our auditing algorithm \ereduce, which is an
enhancement of its plaintext counterpart \reduce \cite{GJD11}. 
We choose to enhance \reduce as it supports a rich set of policies
including HIPAA and GLBA. Further, \reduce has support for
incompleteness in the audit log.  We write $\ereduce^{\pred{DET}}$ to
denote the audit algorithm instance for logs encrypted under
\dt
and $\ereduce^{\pred{KH}}$ for \kh-encrypted audit logs.  
The main difference between \reduce and $\ereduce^{\pred{KH/DET}}$ is
that $\ereduce^{\pred{KH/DET}}$ requires the special \DOPE
data-structure to evaluate the \pred{timeOrder} predicate whereas
\reduce directly checks the linear integer constraint. 
We will describe
$\ereduce^{\pred{KH}}$ in detail and point out the difference between
these two instances.


\subsection{Auxiliary Definitions}

A \emph{substitution} $\sigma$ is a finite map that maps variables to
value, provenance pairs. 
 Each element of the range of a
substitution is of form $\langle v, \ell\rangle$ in which $v$ refers
to the value that the variable is mapped to and $\ell$ refers to the
provenance of $v$. 
 The provenance $\ell$ refers to the
source of the value and is of form $\pred{p}.\pred{a}$ where \pred{p}
represents a table name and \pred{a} represents a column name.
%
We commonly write a substitution $\sigma$ as a finite list of
elements, each element having the form: $\langle x, v^h, v^e,
\ell\rangle$.  For any variable $x$ in $\sigma$'s domain, we use
$\sigma(x).\mathsf{hash}$, $\sigma(x).\mathsf{cipher}$, and
$\sigma(x).\ell$ to select the hash value (\ie, $v^h$), the ciphertext
value (\ie, $v^e$), and the provenance (\ie, $\ell$), 
respectively.

We say
substitution $\sigma_1$ extends $\sigma_2$ (denoted $\sigma_1\geq
\sigma_2$) if $\sigma_1$ agrees with all variable mapping in
$\sigma_2$'s domain.  Given a substitution $\sigma$ we define
$[\sigma]=\{\langle x, \pred{tb}.\pred{cl}\rangle|\exists
v. \sigma(x)=\langle v, \pred{tb}.\pred{cl}\rangle\}$.  
We use
$\sigma\downarrow X$, where $X\subseteq\pred{domain}(\sigma)$,  to denote the substitution $\sigma'$ such that
$\sigma \geq \sigma'$ and domain of $\sigma'$ only contains variables
from the sequence $X$. 
We can lift the $\downarrow$ operation  
for a set of
substitutions $\Sigma$.  We use $\bullet$ to denote the identity
substitution. We say a substitution $\sigma$ \emph{satisfy} a formula
$g$ on a log $e\cL$ if replacing each free variable $x$ in $g$ with concrete value
$\sigma(x).\pred{hash}$ results in a formula that is true on $e\cL$.

\subsection{Algorithm}
Key functions of \ereduce is summarized below.


$\ereduce^\pred{KH}(e\cL,\policy,\vec{\Delta},\sigma)$ is the
  top level function that takes as input an \kh encrypted audit log
  $e\cL$, a constant encrypted policy \policy, a set of tokens
  $\vec{\Delta}$, and an input substitution $\sigma$, and returns a
  residual policy $\psi$. $\psi$ represents a formula 
  containing the part of the original policy \policy that cannot be
  evaluated due to incomplete information in $e\cL$. We use $\bullet$
  as the input substitution to the initial call to
  $\ereduce^\pred{KH}$. \ereduce
  (resp., \sathat and \sat) evaluates the input formula from left to right,
  respecting operator precedence.


$\sathat^\pred{KH}(e\cL,g,\vec{\Delta},\sigma)$
is an auxiliary
  function used by $\ereduce^\pred{KH}$ while evaluating quantifiers
  to get all finite substitutions that satisfy a formula. 
It takes as input an \kh encrypted audit log $e\cL$, a constant
  encrypted formula $g$, a set of tokens $\vec{\Delta}$, and an input
  substitution $\sigma$, and returns all finite substitutions for free
  variables of $g$ that extends $\sigma$ and satisfy $g$ with respect
  to $e\cL$.


$\sat^\pred{KH}(e\cL,\pred{p}(\vec{t}),\vec{\Delta},\sigma)$
is an
  auxiliary function used by $\sathat^\pred{KH}$ for evaluating all
  finite substitutions for a given predicate with respect to an input
  substitution. 
  The inputs $e\cL$, $\vec{\Delta}$, and $\sigma$ have their usual
  meaning. $\pred{p}(\vec{t})$ is a constant encrypted predicate.
  This function returns all finite substitutions for free variables of
  $\pred{p}(\vec{t})$ that extend the input substitution $\sigma$ and
  satisfy $\pred{p}(\vec{t})$ with respect to $e\cL$. The
  implementation of $\sat^\pred{KH}$ is audit log representation
  dependent.  For our case, evaluation of \pred{timeOrder} predicate
  $\sat^\pred{KH}$ consults the \DOPE data structure whereas
  evaluation of other 
  predicates queries the database representing the audit log.

For ease of exposition, we drop the superscript $\pred{KH}$ from
$\ereduce^\pred{KH}$, $\sathat^\pred{KH}$, and $\sat^\pred{KH}$ in the
discussion below.
\ereduce eagerly evaluates as much of the input policy \policy as it can; in case it 
cannot evaluate portions of \policy due to incompleteness in $e\cL$, it returns 
that portion of \policy as part of the result. The return value of \ereduce is thus 
a formula in our logic $\psi$ (\ie, \emph{residual formula}). 
Auditing with \ereduce is an iterative process. When the current log $e\cL$ 
is extended with additional information 
(\ie, removing some incompleteness), resulting in the new log $e\cL_1$ (\viz, $e\cL_1\geq e\cL$), 
one can call \ereduce again with the residual formula $\psi$ as the input policy and $e\cL_1$ 
as the input log. 

We present selected cases of \ereduce in Figure~\ref{fig:ereduce}. 
The conditions above the bar are premises and the condition below
the bar is the conclusion.
We use the notation 
 $\pred{f}(\vec{a})\Downarrow \psi$ to mean function 
 \pred{f} returns $\psi$ when applied to arguments $\vec{a}$.


\begin{figure}[tb]
\centering
\(\small
\begin{array}{l}
\inferrule*[Right=R-P]{
p(\vec{t'})\leftarrow\forall t_i\in\pred{Var}.p(\vec{t})[t_i\mapsto\langle\pred{Adjust}(\sigma(t_i).\pred{hash},\\\\\Delta_{\sigma(t_i).\ell\rightarrow p.i}),\sigma(t_i).\pred{cipher}\rangle]
\\\\
P\leftarrow e\cL(\pred{p}(\vec{t'}))
 }
{\ereduce(e\cL,\pred{p}(\vec{t}),\vec{\Delta},\sigma)\Downarrow P}\\[1ex]
\inferrule*[Right=R-$\bigvee$]
{
\ereduce(e\cL,\policy_1,\vec{\Delta},\sigma)\Downarrow \policy'_1 \\\\
\ereduce(e\cL,\policy_2,\vec{\Delta},\sigma)\Downarrow \policy'_2 \quad 
\psi\leftarrow\policy'_1\vee\policy'_2
}
{
\ereduce(e\cL,\policy_1\vee\policy_2,\vec{\Delta},\sigma)\Downarrow \psi
}\\[2ex]
\inferrule*[Right=R-$\forall$]{
\sathat(e\cL,g,\vec{\Delta},\sigma)\Downarrow \Sigma'\\\\
\forall\sigma_i\in\Sigma'.\ereduce(e\cL,\policy,\vec{\Delta},\sigma_i)\Downarrow \policy_i\\\\
\policy'\leftarrow\forall\vec{x}.(g\wedge\vec{x}\notin[\Sigma'\downarrow\vec{x}]\rightarrow\policy)\\\\
\psi\leftarrow\bigwedge_i\policy_i\wedge\policy'
}
{\ereduce(e\cL,\forall\vec{x}.(g\rightarrow\policy),\vec{\Delta},\sigma)\Downarrow \psi}
\end{array}
\)
\caption{\ereduce description \label{fig:ereduce}}
\end{figure}
%

When the input formula to \ereduce is a predicate $\pred{p}(\vec{t})$
(\textsc{R-P} case), \ereduce uses $\sigma$ and $\vec{\Delta}$ to
replace all variables in $\pred{p}(\vec{t})$ with concrete values
(with proper hash adjustments) to obtain a new ground predicate
$\pred{p}(\vec{t'})$. (A \emph{ground} predicate only has constants as
arguments.)  Then it consults $e\cL$ to check whether
$\pred{p}(\vec{t'})$ exists. If $e\cL( \pred{p}(\vec{t'}) ) =
\mathsf{uu}$, indicating the log doesn't have enough information,
then \ereduce returns $\pred{p}(\vec{t'})$. Otherwise, it returns either true
or false depending on whether there is a row in table \pred{p} with
hash values matching $\vec{t'}$. 
For example, let us assume that \ereduce is called with the
input substitution $\sigma = [\langle p_1, v^h_{k_1}, *,
t.\pred{cl}\rangle, \ldots]$ and the input predicate
$\pred{activeRole}(p_1,\langle \pred{doctor}^h_{k_3}, *\rangle)$
($p_1$ is a variable and \pred{doctor} is a constant). 
$*$ represents ciphertext and is not important for this example.
Let us assume that the column 1 of the
\pred{activeRole} table 
uses the keys $(k_2, *)$ whereas in
$\sigma$, the hash value mapped to $p_1$ is generated using
$k_{1}$. Hence, we have to change the value 
$v^h_{k_1}$ to
$v^h_{k_2}$ using the adjustment key $\Delta_{k_1\mapsto
  k_2}\in\vec{\Delta}$. Then, using the following SQL 
query we check whether a row with the appropriate hash values 
exists: ``select * from \pred{activeRole} where column1Hash=$v^h_{k_2}$ 
and column2Hash=$\pred{doctor}^h_{k_3}$''.
 If such a row exists, then 
$\top$ is returned; otherwise, $\bot$ is returned. 
When \ereduce is
called for \pred{timeOrder}, the same hash adjustment applies before the \DOPE data
structure is consulted.


In rule \textsc{R}-$\bigvee$, \ereduce is recursively called for each
of the
sub-clauses in the disjunction:
$\policy_1$ (resp., for $\policy_2$) and the returned residual formula
is $\policy_1'\vee\policy_2'$.

When the input formula is of form
$\forall\vec{x}(g\rightarrow\policy)$ (\textsc{R}-$\forall$ case), we
first use the function $\sathat$ (described below) to get all
substitutions $\Sigma'$ for $\vec{x}$ that extend $\sigma$ and satisfy
$g$ with respect to $e\cL$. As we require policies to pass EQ mode
check, it is ensured that there is a finite number of such
substitutions for $\vec{x}$. For each of these substitutions
$\sigma_i\in\Sigma'$, we then recursively call \ereduce for \policy to
obtain residual formula $\policy_i$. Then the returned residual
formula is $\bigwedge_i\policy_i \wedge \policy'$ where $\policy'$
ensures that the same substitutions $\sigma_i$ for $\vec{x}$ are not
checked again when $e\cL$ is extended.  

Next, we explain selected rules for \sathat (presented below)
with an example.
%

\hspace{0.1\textwidth}\(\small
\begin{array}{l}
\inferrule*[Right=S-P]{
\Sigma\leftarrow\sat(e\cL, \mathsf{p}(\vec{t}), \vec{\Delta}, \sigma)
}{
\sathat(e\cL, \pred{p}(\vec{t}),\vec{\Delta}, \sigma) \Downarrow \Sigma
}\\[2ex]
\inferrule*[Right=S-$\bigwedge$]
{
\sathat(e\cL,g_1,\vec{\Delta},\sigma)\Downarrow \Sigma'\\\\ 
\forall\sigma_i\in\Sigma'.\sathat(e\cL,g_2,\vec{\Delta},\sigma_i)\Downarrow \Sigma_i
}
{
\sathat(e\cL,g_1\wedge g_2,\vec{\Delta},\sigma)\Downarrow \displaystyle\bigcup_i\Sigma_i
}
\end{array}
\)

%
%
Let us assume \sathat is called with the formula $g\equiv p(x)\wedge
q(x,y)$ and substitution $\sigma=\emptyset$ (empty) as input. The
\textsc{S}-$\wedge$ rule applies and first \sathat is
recursive called on $p(x)$ and $\sigma=\emptyset$. Now, the rule
\textsc{S-P} applies. 
Here, $x$ is not in the domain of 
$\sigma$, so the \sat function consults $e\cL$ (\ie,
using SQL query like: ``select * from $p$'') to find concrete
values of $x$ to make $p(x)$ true. Let us assume that we get 
$\langle v^h_{k_1},*\rangle$ (\ie, $k_1$ is used to hash the
column 1 of table $p$).  Hence, \sat returns the substitution
$\sigma_1=[\langle x,v^h_{k_1},*,p.1\rangle]$ as output. 
Going back to the \textsc{S}-$\wedge$  rule, now the second
premise of \textsc{S}-$\wedge$ calls \sathat for
$q(x,y)$ with each substitution obtained after evaluating $p(x)$, in
our case, $\sigma_1$. Let us assume that column 1 and 2 of table $q$
is hashed with key $k_2$ and $k_3$, respectively. While evaluating,
$q(x,y)$ with $\sigma_1$, \textsc{S-P} rule is used.
$\sigma_1$ already maps variable $x$ but with key $k_1$, thus \sat
converts $v^h_{k_1}$ to
$v^h_{k_2}$ using token $\Delta_{p.1\rightarrow q.1}$.  It then tries to get concrete values for $y$ (with
respect to given value of $x$) by consulting table $q$ in $e\cL$ using
the following SQL query: ``select column2Hash, column2Cipher from $q$ where
column1Hash=$v^h_{k_2}$''. Assume the SQL query returns $\langle
w^h_{k_3},*\rangle$ for column2 (\ie, $y$), \sat returns the
substitution $[\langle x,v^h_{k_1},*,p.1\rangle, \langle
y,w^h_{k_3},*,q.2\rangle]$. 

\noindent\textbf{Differences between \ereduce under \dt and \kh}
As we have shown, $\ereduce^{\pred{KH}}$ tracks the provenance of
encrypted data value required for audit. This is not required by
$\ereduce^{\pred{DET}}$.  For $\ereduce^\pred{DET}$, the substitution
simply $\sigma$ maps variables to deterministic ciphertext. Further,
in the \textsc{R-P} and \textsc{S-P} rule, no adjustment is needed.  

\subsection{Properties}
We have proved the correctness of both $\ereduce^{\pred{DET}}$ and
$\ereduce^{\pred{KH}}$. We show the theorem for $\ereduce^{\pred{KH}}$
below.  Theorem~\ref{thm:khcorrectness} states that the decrypted
result of \ereduce and the result of \reduce are equal with high
probability. The results may not be equal due to hash collisions.  
 Due to space requirement we do not show the proof  here. 
Here, $\pred{EncryptSubstitution}^\pred{KH}$ function encrypts a plaintext
 substitution with provenance to an encrypted one and is very similar
 to the $\pred{EncryptPolicyConstants}^\pred{KH}$ function (see Section \ref{section:kh}).\vspace*{-5pt} 

\begin{theorem}[Correctness of $\ereduce^{\pred{KH}}$]\label{thm:khcorrectness}
For all plaintext policies $\policy_\pred{P}$ and $\psi_\pred{P}$, 
for all constant encrypted policies $\policy_\pred{E}$ and $\psi_\pred{E}$, 
for all database schema  \cS, 
for all plaintext audit logs $\cL=\langle\db, \cT\rangle$, 
for all encrypted audit logs $e\cL=\langle e\db, e\cT\rangle$, 
for all plaintext substitution $\sigma_\pred{P}$, 
for all encrypted substitution $\sigma_\pred{E}$, 
for all \inp, 
for all equality scheme \escheme, 
for all security parameter \param,
for all encryption keys \keys, 
for all token list $\vec{\Delta}$, if all of the following holds:   
(1) $\inp\nsat\policy_\pred{P}: \mymap$, 
(2) $[\sigma_\pred{P}]\supseteq\inp$, 
(3) $\keys = \pred{KeyGen}^\pred{KH}(\param,\cS)$, $\vec{\Delta} = \pred{GenerateToken}(\cS,\escheme, \keys)$, 
(4) $e\cL = \pred{EncryptLog}^{\pred{KH}}(\cL, \cS, \keys)$, 
(5) $\policy_\pred{E} = \pred{EncryptPolicyConstants}^{\pred{KH}}(\policy_\pred{P}\\,\keys)$, 
(6) \pred{AKH} key adjustment  is correct, 
(7) $\sigma_\pred{E} = \pred{EncryptSub}\\\pred{stitution}^{\pred{KH}}(\sigma_\pred{P},\keys)$, 
(8) $\psi_\pred{P} = \reduce(\cL,\sigma_\pred{P},\policy_\pred{P})$, 
(9) $\psi_\pred{E} = \ereduce^{\pred{KH}}(e\cL,\policy_\pred{E},\vec{\Delta},\sigma_\pred{E})$, and 
(10) $\psi'_\pred{P} = \pred{DecryptPolicyConst}\\\pred{ants}^{\pred{KH}}(\psi_\pred{E},\keys)$,   
then $\psi_\pred{p} = \psi'_\pred{P}$ with high probability. 
\end{theorem} 

The correctness theorem for $\ereduce^{\pred{DET}}$ states that the decrypted
result of \ereduce and the result of \reduce are equal.
We omit the formal statement of the correctness theorem for
$\ereduce^{\pred{DET}}$.


The security of \ereduce follows directly from Theorems
\ref{thm:security:det} and \ref{thm:security:kh}. As the security
theorem does not restrict what kind of computation the adversary runs
in polynomial time, \ereduce can be viewed as one instance of such
computation. Hence, \ereduce does not leak any additional information.


%% file: mode.tex
We now present  \emph{EQ mode check}.
which extends mode check \cite{modechecking} introduced in logic
programming.
The EQ mode check is a static analysis of the policy that
serves two purposes: (i) It ensures that \ereduce algorithm
terminates for any policy that passes the check and (ii)
It outputs \emph{the equality scheme} \mymap of the policy,
which is used in both \dt and \kh (see Section
\ref{section:functions}). The EQ mode check runs in the linear
time of the size of the policy. The EQ mode check extends mode
check described in Garg \etal \cite{GJD11} by additionally carrying provenance
and key-adjustment information.
Next, we introduce modes, then we explain our EQ mode checking rules.

\vspace{5pt}
\noindent\textbf{Mode specification.}
The concept of ``\emph{modes}'' comes from logic programming
\cite{modechecking}.  Let us use the following predicate as an example:
Predicate $\pred{tagged}(m, q,
a)$ is true 
when the message $m$ is tagged with principal $q$'s attribute
$a$. Assuming the number of possible messages in English language is
\emph{infinite}, the number of concrete values for variables $m$, $q$,
and $a$ for which \pred{tagged} holds is also \emph{infinite}.
However, if we are given a concrete message (\ie, concrete value for
the variable $m$), then the number of concrete values for $q$ and $a$
for which \pred{tagged} holds
is finite. 
Hence, we say the predicate \pred{tagged}'s argument position 1 is
the input position (denoted by ``$+$'') whereas the argument positions 2 and 3 are output argument positions (denoted
by ``$-$''). We call such a description of input and output specification of a predicate its \emph{mode specification}.
Mode specification of a given predicate signifies that given concrete
values for variables in the input position, the number of concrete
values for variables in the output position for which the given
predicate holds true is \emph{finite}.  Hence $\pred{tagged}(m^+, q^-,
a^-)$ is a valid mode specification whereas $\pred{tagged}(m^-, q^-,
a^+)$ is not.


\vspace{5pt}\noindent\textbf{EQ mode checking.}
EQ mode check uses the mode specification of predicates to check
whether a formula is well-moded. EQ mode check has
two types of judgements: $\inp\bsat g:\langle\outp,\mymap\rangle$ for
guards, and $\inp\nsat\policy:\mymap$ for policy formulas.
Each element of the sets
\inp and \outp is a pair of form $\langle x, \pred{p}.a\rangle$ which
signifies that we obtained concrete value for the variable $x$ with
provenance $\pred{p}.a$ (\ie,
source of the value).
%

\begin{figure}[tb]
\centering
\(
\begin{array}{c}
\inferrule*[Right=g-Pred]
{\forall k\in I(\pred{p}).t_k\in\var\rightarrow t_k\in\fe(\inp)\\\\
\outp=\inp\cup\displaystyle\bigcup_{j\in O(\pred{p})\wedge
t_j\in\var\wedge t_j\notin\fe(\inp)}
\langle t_j, \pred{p}.j\rangle\\\\
\scriptstyle
\mymap=\{\langle \pred{p}'.i, \pred{p}.l\rangle|
0< l\leq\arity{p}\wedge t_l\in\var\wedge\langle t_l,\pred{p}'.i\rangle\in\inp
\}
}{\inp\bsat\pred{p}(t_1,\ldots,t_n):\langle\outp,\mymap\rangle}  \\[2ex]
\inferrule*[Right=g-Conj]
{\inp\bsat g_1: \langle\chi,\mymap_1\rangle\\ \chi\bsat g_2:\langle\outp,\mymap_2\rangle
}{\inp\bsat g_1\wedge g_2: \langle\outp,\mymap_1\cup\mymap_2\rangle}
\\[1ex]
\inferrule*[Right=g-Disj]
{\inp\bsat g_1: \langle\chi_1,\mymap_1\rangle\\ \inp\bsat g_2:\langle\chi_2,\mymap_2\rangle
}{\inp\bsat g_1\vee g_2: \langle\chi_1\Cap\chi_2,\mymap_1\cup\mymap_2\rangle}
\end{array}
\)
\vspace*{-0.1in}
\caption{Selected $\inp\bsat g:\langle\outp,\mymap\rangle$ judgements~\label{fig:bsat}}
\end{figure}

The top level judgement $\inp\nsat\policy:\mymap$ states that given
ground values for variables in set \inp, the formula \policy is
\emph{well-moded} and that audit $\policy$ would require the equality
checking for column pairs given by \mymap.  We call a given policy
\policy \emph{well-moded} if there exists a \mymap for which we can
derive the following judgement: $\{\}\nsat\policy:\mymap$.  The
judgement \nsat uses \bsat as a sub-judgement in the quantifier case
and we explain \bsat first.  The judgement $\inp\bsat
g:\langle\outp,\mymap\rangle$ states that: Given concrete values for
variables in the set \inp, the number of concrete values for variables
in the set \outp (\outp is a subset of the free variables of $g$) for
which the formula $g$ holds true is \emph{finite}. It also outputs the
column pairs which may be checked for equality during evaluation of
$g$.

Selected mode checking rules for guards are listed in
Figure~\ref{fig:bsat}. We explain these rules using an example.
We show how to check
the following formula $g\equiv (p(x^-)\vee q(x^-,z^-))\wedge
r(x^+, y^-)$ with $\inp=\{\}$ 
($\inp=\{\}$ signifies that we do not have concrete values for any variables
yet). The function $I$ (resp., $O$) takes as input a predicate $p$ and
returns all input (resp., output) argument positions of $p$. For
instance, $I(r) = \{1\}$ and $O(r) = \{2\}$.

First, \textsc{G-CONJ} rule
($g^c_1\wedge g^c_2$) applies.
The first premise of \textsc{G-CONJ} requires that  $g^c_1\equiv(p(x^-)\vee q(x^-,z^-))$
is well moded with $\inp=\{\}$. Now,
\textsc{G-DISJ} ($g^d_1\vee g^d_2$) rule can be used to
check the well modedness of $g^c_1$ with $\inp=\{\}$. 
The first and second premise require $g^d_1\equiv p(x^-)$ and $g^d_2\equiv q(x^-,z^-)$ to be independently well moded with
the input $\inp=\{\}$. While evaluating $p(x^-)$ with $\inp=\{\}$ we see the \textsc{G-PRED} rule is applicable. The first premise checks whether
all input variables of $p$, none in this case, are included in \inp; this is trivially satisfied. We use an auxiliary function \fe for checking this,
defined as follows: $\fe(\inp) = \{x\mid\exists \pred{p}, i. \langle x, \pred{p}.i\rangle\in\inp\}$. After evaluating $p$, we will get
concrete values for variable(s) in output  position of $p$ (\ie, $x$ in this case with provenance $p.1$),
hence $\outp=\{\langle x, p.1\rangle\}$, which is expressed in premise 2. Finally, we did not have concrete values for $x$ in \inp, so
we would not need any comparison, hence $\mymap=\{\}$ (premise 3). Similarly, we can derive, $\{\}\bsat  q(x^-,z^-):\langle\{\langle x,q.1\rangle,\langle z,q.2\rangle\},
\{\} \rangle$. Once we have determined both $g^d_1$ and $g^d_2$ are well moded, we see that we are only guaranteed to have
concrete value for variable $x$ (if $g^d_1$ is true we will not get concrete value for $z$) but $x$ can either have provenance $p.1$ or $q.1$.
We have to keep track of both, which is captured using the $\Cap$ operator defined as follows:
$\chi_1\Cap\chi_2=\{
\langle x, p_1.a_1\rangle
\mid
\exists p_2,a_2.((\langle x, p_1.a_1\rangle\in\chi_1 \wedge \langle x, p_2.a_2\rangle\in\chi_2)\bigvee(\langle x, p_1.a_1\rangle\in\chi_2 \wedge \langle x, p_2.a_2\rangle\in\chi_1)).
\}$. So we have, $\{\}\bsat g^c_1:\langle \{\langle
x,p.1\rangle,\langle x,q.1\rangle\} ,\{\}\rangle$.

Now let's go back to the second premise of
\textsc{G-CONJ}, which requires that  $r(x^+, y^-)$ is well-moded with respect to $\chi=\{\langle x,p.1\rangle,\langle x,q.1\rangle\}$. The \textsc{G-PRED} rule
is applicable again for checking this. The first premise, requiring variables in input argument position (\ie, $x$ in this case) are
given by $\inp$, is satisfied. According to the second premise, we will additionally get concrete value for $y$ with provenance $r.2$ hence
$\outp=\{\langle x,p.1\rangle,\langle x,q.1\rangle, \langle y,r.2\rangle\}$. Finally, concrete value for $x$ (with provenance $p.1$ and $q.1$)
is needed while evaluating $r$ ($x$ is in input argument position 1 of $r$), hence we need to check for equality between the following column pairs, $p.1$, $r.1$ and
$q.1$, $r.1$, \ie, $\mymap=\{\langle p.1,r.1\rangle,\langle q.1,r.1\rangle\}$.  The directionality of the pairs, \eg, $\langle p.1,r.1\rangle$,
matters for efficiency. For instance, we are given a concrete value  hashed with key $k_{p.1}$ for variable $x$. While evaluating $r$ we would need concrete
value of $x$ ($x$ is in input argument position of $r$). Now, if we are given the token $\Delta_{r.1\rightarrow p.1}$ we cannot directly evaluate $r$  without incurring
additional computational overhead.

Top-level mode checking rules for policy formulas are very similar to
those for guards, except that formulas do not ground variables.  We
show the rule for universal quantification below.  The audit algorithm
checks formulas of form $\forall\vec{x}.(g\rightarrow\policy)$ by
first obtaining finite number of substitutions for $\vec{x}$ that satisfy
$g$ and then checking whether \policy holds true for each of these
substitution.
\[
\inferrule*[Right=Univ]{\inp\bsat g:\langle\outp,\mymap_g\rangle\\ \vec{x}\subseteq\fe(\outp)\\\\
fv(g)\subseteq\fe(\inp)\cup\{\vec{x}\} \\ \outp\nsat\policy:\mymap_c}
{\inp\nsat\forall\vec{x}.(g\rightarrow\policy):\mymap_g\cup\mymap_c}
\]
The first premise of \textsc{UNIV} 
checks that we have only finite number of substitutions for $\vec{x}$ that satisfy $g$ with respect to
\inp and with equality scheme $\mymap_g$. This is necessary for
termination while checking universal formulas as the domain of the
variables can be
infinite. We then check whether we have substitutions for all quantified variables $\vec{x}$ and also all the free variables
of $g$. Finally, we inductively check whether \policy is well-moded with respect to the ground variables we obtained while
evaluating $g$ with equality scheme $\mymap_c$. Then, the resulting equality scheme is $\mymap_g\cup\mymap_c$.

%% file: implementation2.tex
\begin{figure}[b!]
\centering
\includegraphics[width=0.40\textwidth]{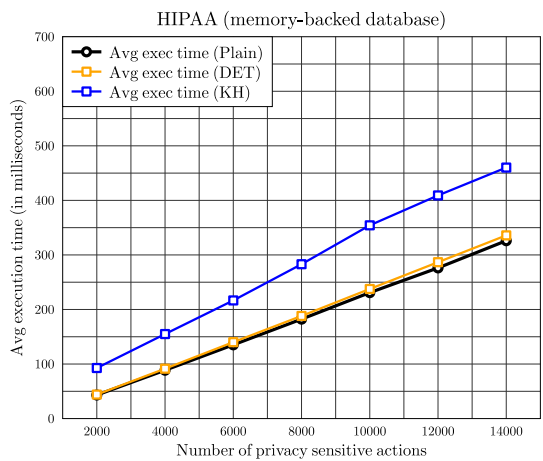}
\vspace*{-0.1in}
\caption{\textbf{Experimental results HIPAA}\label{fig:memory-timing}}
\end{figure}
We report on our empirical evaluation of \ereduce on the \dt and \kh
schemes. We run experiments on a 2.67GHz Intel Xeon X5650 CPU with
Debian Linux 7.6 and 50GB of RAM, of which no more than 3.0 GB is
used. SQLite version 3.8.7.1 is used to store the plaintext and
encrypted logs. We aggressively index all database columns in input
argument positions as specified in mode specifications. In \dt, the
index is built over deterministically encrypted values; in \kh, the
index is built over hashed values. For deterministic encryption, we
use AES with a variation of the CMC mode~\cite{CMC03} with a fixed IV
and a 16 byte block size. We use 256 bit keys. For the AKH scheme, we
use the library by Popa~\etal~\cite{cryptDB}. The underlying elliptic
curve is the NIST-approved NID\_X9\_62\_prime192v1. We use privacy
policies derived from the GLBA and HIPAA privacy rules and cover 4 and
13 representative clauses of these rules, respectively.

We use synthetically generated plaintext audit logs. Given an input
policy and a desired number of privacy sensitive actions, our audit
log generation algorithm randomly decides whether each action will be
policy compliant or not. To generate log entries for a compliant
action, the algorithm traverses the abstract syntax tree of the policy
and generates 
instances of atoms that together satisfy the
policy. For the non-compliant actions, we randomly choose atoms to
falsify a necessary condition.  Our synthetic log generator also
outputs the \DOPE data structure but with plaintext values for
timestamps.  We generate logs with 2000 to 14000 privacy-sensitive
actions. Each plaintext log is encrypted with both the \dt and \kh
schemes.  The maximum plaintext audit log size we considered is 17
MB. The corresponding maximum encrypted log sizes in \dt and \kh are
67.3MB and 267MB, respectively. Most of the size of \kh-encrypted log
comes from the keyed hashes.

We measure the relative overhead of running \ereduce on logs encrypted
with \dt and \kh, choosing \reduce on plaintext audit log as the
baseline. We experiment with both RAM-backed and disk-backed versions
of SQLite. We report here only the memory-backed results
(the disk-backed results
are similar). Figure~\ref{fig:memory-timing} shows the average
execution time per privacy-sensitive action for the HIPAA
policy in all three configurations (GLBA results are similar). The number of privacy-sensitive
actions (and, hence, the log size) varies on the x-axis.
The overhead of
\dt is very low, between 3\% and 9\%. This is unsurprising, because
no cryptographic operations are needed during audit. The
overhead comes from the need to read and compare longer
(encrypted) fields in the log and from having to use the \DOPE data
structure. With \kh, overheads are much higher, ranging from 63\% to
406\%. These overheads come entirely from two sources: the cost of
reading a much larger database and the cost of performing hash
adjustments to check equality of values in different columns.
We observe that the overhead due to the
increased database size is more than that due to hash adjustment. For
the policies we experimented with, the per-action overhead due to
database size grows linearly, but the overhead due to hash adjustments
is relatively constant. There is substantial room (\ie, 30\% of the total overhead incurred by \kh) 
for improving the
efficiency of \ereduce on \kh, \eg, by caching previous key-adjustments, which we
currently do not do.

%% file: related1.tex
In this section, we briefly review the existing work that are most relevant to our
approach.

\noindent
\textbf{Functional encryption:}
Function encryption scheme~\cite{BSW11,GPSW06,LOSTW10,neil10} enables one (possibly with the possession of some tokens)
to calculate a function value over encrypted arguments. The output of the
function is in plaintext unlike homomorphic encryption~\cite{Gentry09}.  Our approach can be viewed as
attempting to mimic functional encryption where the function is compliance checking of a given policy.
We have the following differences with functional encryption: (i) some of the output can be encrypted in our case,
(ii) we have multiple arguments each of which can be encrypted with different keys,
 and (iii) also we cannot hide the function being computed as the policy is
public knowledge. However, Goldwasser~\etal~\cite{Goldwasser14} introduce multi-input functional encryption in which
multiple arguments are considered.

\noindent
\textbf{Predicate encryption:}
\emph{Property-preserving encryption}~\cite{PR12} or \emph{predicate encryption}~\cite{SSW09,BW07,KSW08}
can be viewed as a special case of the functional encryption where the function returns boolean value.
The compliance checking can be viewed as a variation of predicate encryption where the predicate outputs
`0' for violation of the policy and `1' for satisfaction. Traditionally, predicate encryption schemes consider
arguments encrypted with a single encryption key whereas in our case we have arguments encrypted with multiple encryption keys.
Pandey and Rouselakis~\cite{PR12}
present several notions of security for symmetric predicate encryption and
constructs one such encryption scheme for the checking the orthogonality of two encrypted vectors ($\vec{x}\cdot\vec{y} \overset{?}{=} 0\mbox{ mod }p$).
Our security definition \ECPLA is inspired by their
\pred{LoR} security notion. Shen~\etal~\cite{SSW09} present a symmetric predicate encryption scheme which supports
  inner product queries. Their approach can also be used for equality checking but using it will result in the database indexing
to be unusable in our case. They introduce a notion of security called \emph{predicate privacy}. Our approach
cannot provide \emph{predicate privacy} as the policy is known to the adversary.


\noindent\textbf{Structured encryption:}
Chase and Kamara~\cite{CK11} introduce structured encryption
for structured data and which maintains the structure of the data after the encryption.
They construct encryption schemes for graph structures and allow adjacency queries, neighboring queries,
and also focused subgraph queries. Our encryption of the audit logs could be viewed as an instance of structured encryption
where the queries we are interested in, are relevant to compliance checking with respect to a given policy.

\noindent\textbf{Searchable audit log:}
Waters~\etal~\cite{WBDS04} present a framework in which they allow both confidentiality and integrity protection
with the ability to search the encrypted audit logs based on some keywords. They use hash chains for integrity protection
and use identity-based encryption~\cite{BF03} with extracted keywords to enable searching the audit log~\cite{BCOP03}.
In our current work, we only consider confidentiality of the data and assume existence of complementary techniques to
ensure integrity of the audit log~\cite{SK98,SK99,KS99,Holt06}. The policies against which we check the audit log is more expressive than what they consider.
Additionally, we require time-stamp comparison which is their framework does not allow.

\noindent\textbf{Searchable encrypted audit log:} Waters~\etal
present a framework in which they provide both confidentiality and
integrity protection with the ability to search the encrypted audit
logs based on keywords~\cite{WBDS04}. They use 
use identity-based encryption~\cite{BF03} with
extracted keywords to enable searching the audit log~\cite{BCOP03}.
We only consider confidentiality of the data and
assume existence of complementary techniques to ensure integrity of
the audit log~\cite{SK99}. The policies we consider are more
expressive than theirs and we support time-stamp
comparison which is not present in their framework.

\noindent \textbf{Order-preserving encryption:}
Boldyreva~\etal~\cite{BCLO} present a symmetric encryption scheme that
maintains the order of the plaintext data which does not satisfy the
ideal \OCPA security definition.  Popa~\etal present the
\MOPE scheme which we enhance to support timestamp comparison with
displacements~\cite{mope}.  Recently, Kerschbaum and Schr\"{o}pfer
present a keyless order preserving encryption scheme for outsourced
data~ \cite{KS14}. In their approach, the owner of the plaintext data is required
to keep a dictionary of mapping from plaintext to ciphertext which is
undesirable in our scenario.

\noindent
\textbf{Querying outsourced database:}
Hacig\"{u}m\"{u}\c{s} \etal \cite{HILM02} develop a system that allows
querying over encrypted data. Their strategy is to ask the client to
decrypt data to enable
operations on the encrypted data.
Tu \etal \cite{TKMZ13} introduces split client/server query execution
for processing analytical queries on encrypted databases.
\framework does not require any query processing on the
client side as it is untrusted.
Damiani \etal \cite{DVJPS03} developed a secure indexing approach for
querying an encrypted database. We do not require modification to the
indexing algorithm of the DBMS. 

\emph{cryptDB} developed by Popa~\etal~\cite{cryptDB} allows queries
over encrypted databases. Their goal is to leave the
front-end of the application running on the client side, untouched and
relying on a trusted proxy to perform appropriate encryptions and
decryptions of the queries and the results.  Our setting is different
in that our application runs in the cloud without requiring the
existence of a trusted proxy.  Unlike cryptDB, \framework supports a
restricted set of SQL queries for auditing. We prove
end-to-end security guarantees of our algorithm, which has not been
done for cryptDB.

\noindent \textbf{Privacy policy compliance checking:}
Prior work on logic-based compliance checking algorithms focus on plaintext logs
~\cite{BKMZ12,BKMZ13,bauer13,CJDD14,GJD11}. We are the first to use
logic-based approach to check encrypted logs for policy compliance.

\noindent\textbf{Auditing retention policies:}
Lu~\etal~\cite{LMI13} presents a framework for auditing the changes to database with respect to a retention policy.
They also consider retention history being incomplete. However, they work on plaintext audit log. To support their
queries \framework needs to be enhanced.

%% file: conclusion.tex
We presented an auditing algorithm that checks compliance over
encrypted audit logs with respect to an expressive class of privacy
policies.  We introduced a novel notion of audit log equivalence that
enables us to obtain an end-to-end security definition that precisely
captures the information leakage during the auditing process. We then
prove secure, two instances of the auditing algorithm which differ in
the encryption scheme, under this definition.
Empirical evaluation demonstrates that both instances of our algorithm
have low to moderate overhead compared to a baseline algorithm that
only supports plaintext audit logs.


%% file: TR.bbl
\begin{thebibliography}{10}

\bibitem{ANB98}
H.~Andr\'{e}ka, I.~N\'{e}meti, and J.~van Benthem.
\newblock Modal languages and bounded fragments of predicate logic.
\newblock {\em Journal of Philosophical Logic}, 27(3):217--274, 1998.

\bibitem{modechecking}
K.~Apt and E.~Marchiori.
\newblock Reasoning about prolog programs: From modes through types to
  assertions.
\newblock {\em Formal Aspects of Computing}, 1994.

\bibitem{askarov07}
A.~Askarov and A.~Sabelfeld.
\newblock Gradual release: Unifying declassification, encryption, and key
  release policies.
\newblock In {\em IEEE Security and Privacy}, 2007.

\bibitem{BS11}
S.~Bajaj and R.~Sion.
\newblock Trusteddb: A trusted hardware based database with privacy and data
  confidentiality.
\newblock In {\em ACM SIGMOD '11}.

\bibitem{BKMZ13}
D.~Basin, F.~Klaedtke, S.~Marinovic, and E.~Zalinescu.
\newblock Monitoring of temporal first-order properties with aggregations.
\newblock In {\em RV}, 2013.

\bibitem{BKMZ12}
D.~Basin, F.~Klaedtke, S.~Marinovic, and E.~Z\u{a}linescu.
\newblock Monitoring compliance policies over incomplete and disagreeing logs.
\newblock In {\em Runtime Verification}. 2013.

\bibitem{bauer13}
A.~Bauer, J.-C. K\"{u}ster, and G.~Vegliach.
\newblock From propositional to first-order monitoring.
\newblock In {\em Runtime Verification}. 2013.

\bibitem{BCLO}
A.~Boldyreva, N.~Chenette, Y.~Lee, and A.~O'Neill.
\newblock Order-preserving symmetric encryption.
\newblock In {\em EUROCRYPT '09}.

\bibitem{BCOP03}
D.~Boneh, G.~D. Crescenzo, R.~Ostrovsky, and G.~Persiano.
\newblock Public key encryption with keyword search.
\newblock Cryptology ePrint Archive, Report 2003/195.

\bibitem{BF03}
D.~Boneh and M.~Franklin.
\newblock Identity-based encryption from the {Weil} pairing.
\newblock {\em SIAM J. of Computing}, 2003.

\bibitem{BSW11}
D.~Boneh, A.~Sahai, and B.~Waters.
\newblock Functional encryption: Definitions and challenges.
\newblock In {\em Proceedings of the 8th Conference on Theory of Cryptography},
  TCC'11, pages 253--273, Berlin, Heidelberg, 2011. Springer-Verlag.

\bibitem{BW07}
D.~Boneh and B.~Waters.
\newblock Conjunctive, subset, and range queries on encrypted data.
\newblock In S.~Vadhan, editor, {\em Theory of Cryptography}, volume 4392 of
  {\em Lecture Notes in Computer Science}, pages 535--554. Springer Berlin
  Heidelberg, 2007.

\bibitem{CK11}
M.~Chase and S.~Kamara.
\newblock Structured encryption and controlled disclosure.
\newblock Cryptology ePrint Archive, Report 2011/010, 2011.
\newblock \url{http://eprint.iacr.org/}.

\bibitem{CJDD14}
O.~Chowdhury, L.~Jia, D.~Garg, and A.~Datta.
\newblock Temporal mode-checking for runtime monitoring of privacy policies.
\newblock In {\em CAV '14}.

\bibitem{DVJPS03}
E.~Damiani, S.~D.~C. Vimercati, S.~Jajodia, S.~Paraboschi, and P.~Samarati.
\newblock Balancing confidentiality and efficiency in untrusted relational
  dbmss.
\newblock In {\em CCS '03}.

\bibitem{DGJKD10}
H.~DeYoung, D.~Garg, L.~Jia, D.~Kaynar, and A.~Datta.
\newblock Experiences in the logical specification of the hipaa and glba
  privacy laws.
\newblock In {\em WPES '10}.

\bibitem{GJD11}
D.~Garg, L.~Jia, and A.~Datta.
\newblock Policy auditing over incomplete logs: Theory, implementation and
  applications.
\newblock In {\em CCS '11}.

\bibitem{Gentry09}
C.~Gentry.
\newblock Fully homomorphic encryption using ideal lattices.
\newblock In {\em STOC '09}.

\bibitem{Goldwasser14}
S.~Goldwasser, S.~Gordon, V.~Goyal, A.~Jain, J.~Katz, F.-H. Liu, A.~Sahai,
  E.~Shi, and H.-S. Zhou.
\newblock Multi-input functional encryption.
\newblock In {\em EUROCRYPT '14}.

\bibitem{GPSW06}
V.~Goyal, O.~Pandey, A.~Sahai, and B.~Waters.
\newblock Attribute-based encryption for fine-grained access control of
  encrypted data.
\newblock In {\em CCS '06}.

\bibitem{HILM02}
H.~Hacig\"{u}m\"{u}\c{s}, B.~Iyer, C.~Li, and S.~Mehrotra.
\newblock Executing sql over encrypted data in the database-service-provider
  model.
\newblock In {\em SIGMOD '02}.

\bibitem{CMC03}
S.~Halevi and P.~Rogaway.
\newblock A tweakable enciphering mode.
\newblock In {\em CRYPTO '03}.

\bibitem{HIPAA}
{Health and Human Services}.
\newblock Health insurance portability and accountability act, 1996.
\newblock Public Law 104-191.

\bibitem{Holt06}
J.~E. Holt.
\newblock Logcrypt: Forward security and public verification for secure audit
  logs.
\newblock In {\em ACSW Frontiers '06}.

\bibitem{KSW08}
J.~Katz, A.~Sahai, and B.~Waters.
\newblock Predicate encryption supporting disjunctions, polynomial equations,
  and inner products.
\newblock In {\em Proceedings of the Theory and Applications of Cryptographic
  Techniques 27th Annual International Conference on Advances in Cryptology},
  EUROCRYPT'08, pages 146--162, Berlin, Heidelberg, 2008. Springer-Verlag.

\bibitem{KS99}
J.~Kelsey and B.~Schneier.
\newblock Minimizing bandwidth for remote access to cryptographically protected
  audit logs.
\newblock In {\em Recent Advances in Intrusion Detection}, 1999.

\bibitem{KS14}
F.~Kerschbaum and A.~Schroepfer.
\newblock Optimal average-complexity ideal-security order-preserving
  encryption.
\newblock In {\em CCS '14}.

\bibitem{LOSTW10}
A.~Lewko, T.~Okamoto, A.~Sahai, K.~Takashima, and B.~Waters.
\newblock Fully secure functional encryption: Attribute-based encryption and
  (hierarchical) inner product encryption.
\newblock In H.~Gilbert, editor, {\em Advances in Cryptology – EUROCRYPT
  2010}, volume 6110 of {\em Lecture Notes in Computer Science}, pages 62--91.
  Springer Berlin Heidelberg, 2010.

\bibitem{LMI13}
W.~Lu, G.~Miklau, and N.~Immerman.
\newblock Auditing a database under retention policies.
\newblock {\em The VLDB Journal}, 22(2):203--228, Apr. 2013.

\bibitem{neil10}
A.~O'Neill.
\newblock Definitional issues in functional encryption.
\newblock Cryptology ePrint Archive, Report 2010/556, 2010.
\newblock \url{http://eprint.iacr.org/}.

\bibitem{PR12}
O.~Pandey and Y.~Rouselakis.
\newblock Property preserving symmetric encryption.
\newblock In {\em EUROCRYPT'12}.

\bibitem{mope}
R.~A. Popa, F.~H. Li, and N.~Zeldovich.
\newblock An ideal-security protocol for order-preserving encoding.
\newblock In {\em IEEE S\&P}, 2013.

\bibitem{cryptDB}
R.~A. Popa, C.~M.~S. Redfield, N.~Zeldovich, and H.~Balakrishnan.
\newblock Cryptdb: Protecting confidentiality with encrypted query processing.
\newblock In {\em SOSP '11}.

\bibitem{ADJJOIN}
R.~A. Popa and N.~Zeldovich.
\newblock Cryptographic treatment of cryptdb's adjustable join,
  {MIT-CSAIL-TR-2012-006}.

\bibitem{SK98}
B.~Schneier and J.~Kelsey.
\newblock Cryptographic support for secure logs on untrusted machines.
\newblock In {\em USENIX Security Symposium '98}.

\bibitem{SK99}
B.~Schneier and J.~Kelsey.
\newblock Secure audit logs to support computer forensics.
\newblock {\em {ACM} Trans. Inf. Syst. Secur. '99}.

\bibitem{GLBA}
{Senate Banking Committee}.
\newblock {G}ramm-{L}each-{B}liley {A}ct, 1999.
\newblock Public Law 106-102.

\bibitem{SSW09}
E.~Shen, E.~Shi, and B.~Waters.
\newblock Predicate privacy in encryption systems.
\newblock In {\em Proceedings of the 6th Theory of Cryptography Conference on
  Theory of Cryptography}, TCC '09, pages 457--473, Berlin, Heidelberg, 2009.
  Springer-Verlag.

\bibitem{SWP00}
D.~X. Song, D.~Wagner, and A.~Perrig.
\newblock Practical techniques for searches on encrypted data.
\newblock In {\em IEEE Security \& Privacy '00}.

\bibitem{TKMZ13}
S.~Tu, M.~F. Kaashoek, S.~Madden, and N.~Zeldovich.
\newblock Processing analytical queries over encrypted data.
\newblock In {\em PVLDB'13}.

\bibitem{WBDS04}
B.~R. Waters, D.~Balfanz, G.~Durfee, and D.~K. Smetters.
\newblock Building an encrypted and searchable audit log.
\newblock In {\em NDSS '04}.

\end{thebibliography}
